\documentclass[unsortedaddress,superscriptaddress,pra,notitlepage]{revtex4}

\usepackage{amsmath,amssymb}
\usepackage{braket}
\usepackage{theorem}
\usepackage{color}
\usepackage{amsfonts}
\usepackage[mathscr]{eucal}
\usepackage[hidelinks]{hyperref}
\usepackage{graphicx}
\usepackage{xcolor}
\usepackage{tikz}
\usepackage{scalerel}
\usepackage{bm}

\usetikzlibrary{calc,decorations.pathreplacing}
\usetikzlibrary{arrows,shapes}

\newtheorem{definition}{Definition}[section]

\newtheorem{thm}[definition]{Theorem}
\newtheorem{prop}[definition]{Proposition}
\newtheorem{lemma}[definition]{Lemma}

\newtheorem{rem}[definition]{Remark}

\newtheorem{cor}[definition]{Corollary}

\newenvironment{proof}[1][Proof]{\begin{trivlist}
\item[\hskip \labelsep {\bfseries #1}]}{\hfill$\Box$\end{trivlist}}

\newcommand{\nn}{\nonumber \\}

\def\theta{\vartheta}
\def\hil{{\mathcal H}}
\def\kil{{\mathcal K}}

\def\B{{\mathcal B}}

\def\I{{\mathcal I}}

\def\M{\mathcal{M}}

\def\S{{\mathcal S}}

\def\X{{\mathcal X}}

\def\imp{\implies}

\def\ep{\varepsilon}
\def\bN{\mathbb{N}}
\def\bC{\mathbb{C}}

\def\bR{\mathbb{R}}

\def\bz{\left(}
\def\jz{\right)}

\def\inv{^{-1}}

\def\egy{\mathbf 1}
\def\map{\Phi}

\def\W{W}

\def\what{\widehat}
\def\oll{\overline}

\def\rho{\varrho}
\def\povm{\mathrm{POVM}}

\def\cl{\mathrm{cl}}

\def\nn{\nonumber}

\def\meas{\mathrm{meas}}

\def\bary{\mathrm{b}}

\def\p{_{\ge 0}}
\def\pne{_{\gneq 0}}
\def\pp{_{>0}}

\def\valt{\cdot}

\def\qv{\mathbf{q}}

\def\Um{\mathrm{Um}}

\def\req{Q}
\def\LE{\mathrm{LE}}
\def\GM{\mathrm{G}}
\def\wtilde{\widetilde}
\def\hell{\omega}
\def\chtimes{\otimes_{\mathrm{ch}}}
\def\univ{\Omega}
\def\symm{\mathrm{sym}}

\newcommand{\ki}[1]{\emph{#1}}

\newcommand{\ds}{\mbox{ }\mbox{ }}

\newcommand{\norm}[1]{\left\| #1\right\|}

\newcommand{\inner}[2]{\left\langle #1 , #2\right\rangle}
\newcommand{\diad}[2]{\left|#1\right\rangle\!\left\langle #2\right|}

\newcommand{\pr}[1]{\diad{#1}{#1}}

\newcommand{\ch}[1]{W_{#1}}

\newcommand{\DP}[1]{D^{\bary,\qv}_P\bz #1\jz}

\newcommand{\acc}[2]{#1\vert_{#2}}

\newcommand{\jsp}[2]{S_{#1,#2}}
\newcommand{\lreg}[1]{\underline{#1}}
\newcommand{\ureg}[1]{\overline{#1}}
\newcommand{\breg}[1]{\overline{\underline{#1}}}

\makeatletter
\renewcommand{\p@enumii}{}
\makeatother

\DeclareMathOperator{\Tr}{Tr}

\DeclareMathOperator{\supp}{supp}

\DeclareMathOperator{\ran}{ran}
\DeclareMathOperator{\dom}{dom}

\DeclareMathOperator{\spec}{spec}

\DeclareMathOperator{\logn}{\widehat\log}

\DeclareMathOperator{\divv}{\Delta}

\DeclareMathOperator{\D}{\mathit{D}}
\DeclareMathOperator{\DU}{\mathit{D}^{Um}}
\DeclareMathOperator{\DBS}{\mathit{D}^{\max}}

\DeclareMathOperator{\nlog}{\what\log}

\DeclareMathOperator*{\medcup}{\scalerel*{\cup}{\textstyle\sum}}
\DeclareMathOperator*{\medwedge}{\scalerel*{\wedge}{\textstyle\sum}}

\begin{document}

\title{Regularized barycentric R\'enyi divergences}

\author{Mil\'an Mosonyi}
\email{milan.mosonyi@gmail.com}

\affiliation{Department of Analysis and Operations Research, Institute of Mathematics,
Budapest University of Technology and Economics,
M\H uegyetem rkp.~3., H-1111 Budapest, Hungary}

\begin{abstract}
\centerline{\textbf{Abstract}}
\vspace{.3cm}

Barycentric R\'enyi divergences were introduced in [Mosonyi, Bunth, Vrana, Linear Algebra and its Applications, 2024] as an alternative to standard Kubo-Ando constructions to define multivariate quantum R\'enyi divergences. 
They are defined via a variational expression and depend on a finite collection of quantum relative entropies
$D^{q_x}$. When all the relative entropies are monotone under CPTP maps then so are the corresponding barycentric R\'enyi divergences, and when all the relative entropies are additive then the corresponding 
barycentric R\'enyi divergences are subadditive under tensor product. Additivity has only been established before for the case where all $D^{q_x}$ are chosen to be the Umegaki relative entropy, which is also the only case where the barycentric R\'enyi divergence (called the minimal one) was known to admit an explicit expression. 

Here we settle the problem of additivity by showing that for any choice of monotone quantum relative entropies, the regularized barycentric R\'enyi divergence coincides with the  minimal barycentric R\'enyi divergence on invertible inputs. This in turn implies that if at least two of the $D^{q_x}$ are strictly larger than the Umegaki relative entropy (in a precise sense), then the resulting barycentric R\'enyi divergence is not even weakly additive.
\end{abstract}

\maketitle

\tableofcontents

\section{Introduction}

For a finite collection $W=(\ch{x})_{x\in\X}$ of probability distributions on a finite set $\I$ and a probability distribution $P$ on $\X$, the corresponding $P$-weighted multivariate R\'enyi divergence
of $W$ is defined as
\begin{align}\label{eq:classical multirenyi def}
D_P^{\cl}(W):=-\log\req_{P}(W),\ds\ds\ds
\req_{P}(W):=\sum_{i\in\I}\prod_{x\in\supp(P)}\ch{x}(i)^{P(x)}\,.
\end{align}
These quantities feature, e.g., in the characterization of the optimal exponent of the symmetric error probability of classical state exclusion \cite{Antidist2023}, as well as in the characterization 
of the convertibility of one finite collection of probability distributions into another \cite{farooq2023asymptotic}.
For the solution of the quantum generalizations of these problems, one needs a suitable extension of 
\eqref{eq:classical multirenyi def} for non-commuting families of quantum states, motivating the study 
of multivariate quantum R\'enyi divergences.

Due to the non-commutativity of quantum states, \eqref{eq:classical multirenyi def} admits infinitely many different quantum generalizations already in the bivariate case $|\X|=2$, for instance, the 
Petz-type \cite{P86}, the sandwiched \cite{Renyi_new,WWY}, and the more general
R\'enyi $(\alpha,z)$-divergences \cite{AD}, the measured and the maximal \cite{Matsumoto_newfdiv} R\'enyi divergences, and more
\cite{CapelRico_fdiv,FawziFawzi2021,mosonyi2022geometric}. In the multivariate case $|\X|>2$, even writing down explicit quantum extensions  
that satisfy some basic requirements like 
additivity under tensor product and data processing inequality, is a highly non-trivial task. Defining 
$\req_P(W)$ as the trace of some of the well-known notions of multivariate matrix geometric means
derived from the Kubo-Ando geometric means
\cite{AndoLiMathias2004,BMP,Bhatia_Holbrook2006,KA,LawsonLim2014,Lim_Palfia2012,Moakher_matrixmean} yields such quantities, although most of them do not admit an explicit expression. Moreover, 
they are known not to reduce to the operationally relevant quantities in the bivariate case, e.g., 
in state discrimination \cite{ANSzV,Hayashicq,Nagaoka,MO}. 

An alternative approach to defining multivariate quantum R\'enyi divergences was put forward
in \cite{mosonyi2022geometric},
inspired by a variational representation of classical R\'enyi divergences. For any collection 
$\D^{\qv}:=(D^{q_x})_{x\in\X}$ of quantum relative entropies (to be defined in Section \ref{sec:relentr}), the corresponding $P$-weighted barycentric R\'enyi divergence of $W$ was defined as
\begin{align*}
D_P^{\bary,\qv}(W):=\inf_{\omega}\sum_{x\in\X}P(x)D^{q_x}(\omega\|\ch{x}),
\end{align*}
where the optimization is over all quantum states $\omega$.
When $D^{q_x}$ is the Umegaki relative entropy \cite{Umegaki} for every $x\in\X$, the above optimization 
can be solved explicitly, with the optimal $\omega$ being 
\begin{align*}
\omega_P^{\bary,\Um}(W)=
\begin{cases}
\frac{\GM_P^{\LE}(W)}{\Tr \GM_P^{\LE}(W)},&S\ne 0,\\
\text{any state},&S=0,
\end{cases},
\end{align*}
where $S$ is the projection onto the joint support of the $W_x$, $x\in\supp P$, and 
\begin{align*}
\GM_P^{\LE}(W):=S\exp\bz\sum_{x\in\X}P(x)S(\log \ch{x})S\jz
\end{align*}
is the $P$-weighted \ki{log-Euclidean geometric mean} of $W$. This then gives
\begin{align}\label{eq:LE Renyidiv}
D_P^{\bary,\Um}(W)=-\log\Tr\GM_P^{\LE}(W).
\end{align}

For a general collection of monotone and (weakly) tensor subadditive quantum relative entropies
$\D^{\qv}$ and every $P$, the corresponding $P$-weighted R\'enyi divergence $D_P^{\bary,\qv}$
is easily seen to satisfy the data processing inequality (i.e., monotonicity under completely positive 
trace-preserving maps), and (weak) subadditivity under tensor products, while the question of additivity was left open in \cite{mosonyi2022geometric}. This latter property would be easy to verify if $D_P^{\bary,\qv}$ could be expressed in some explicit
form; however, no such explicit expression has been established before, except for the case \eqref{eq:LE Renyidiv}.

In this paper we 
identify one more monotone and additive quantum relative entropy for which   
the corresponding $P$-weighted R\'enyi divergence can be computed explicitly. It is defined as
\begin{align*}
D^{\Um,\#_0}(\rho\|\sigma):=\DU(\rho\|\acc{\sigma}{\rho}),
\end{align*}
where $\acc{\sigma}{\rho}$ is the absolutely continuous part of $\sigma$ with respect to $\rho$
\cite{Anderson1971,Ando_Lebesgue}. We show that 
\begin{align*}
D_P^{\bary,\Um,\#_0}(W)=-\log\Tr\GM_P^{\LE}\bz\acc{W}{S}\jz,
\end{align*}
where $\acc{W}{S}=(\acc{W_x}{S})_{x\in\X}$.

In our main result, given in Theorems \ref{thm:main}--\ref{thm:main2} and Proposition \ref{prop:nonadd}, we 
settle the additivity problem for the barycentric R\'enyi divergences as follows:
\begin{thm}
For every finite set $\X$, every 
$P\in\S(\X)$, every 
$D^{\qv}=(D^{q_x})_{x\in\X}$ such that $D^{q_x}$ is monotone for every 
$x\in\supp P$, and every $W\in\B(\X,\hil)\p$, 
\begin{align}
D_P^{\bary,\Um}(W)
&=
\lim_{n\to+\infty}\frac{1}{n}D_P^{\bary,\meas}\bz W^{\otimes n}\jz\\
&\le \liminf_{n\to+\infty}\frac{1}{n}D_P^{\bary,\qv}\bz W^{\otimes n}\jz\\
&\le \limsup_{n\to+\infty}\frac{1}{n}D_P^{\bary,\qv}\bz W^{\otimes n}\jz\\
&\le
\lim_{n\to+\infty}\frac{1}{n}D_P^{\bary,\max}\bz W^{\otimes n}\jz
= D_P^{\bary,\Um,\#_0}(W),
\end{align}
where $D_P^{\bary,\meas}$ and $D_P^{\bary,\max}$ are the $P$-weighted barycentric R\'enyi divergences corresponding to the measured \cite{Donald1986} and the maximal \cite{BS,Matsumoto_newfdiv}
quantum relative entropies, respectively, and $W^{\otimes n}=(W_x^{\otimes n})_{x\in\X}$.
Moreover, all inequalities 
above hold as
equalities  when $W_x^0=W_y^0$ for all $x,y\in\supp P$.
\end{thm}

\begin{thm}
Let $\X$ be a finite set with at least two points, $P$ be a strictly positive probability distribution on $\X$, and for every $x$, let $D^{q_x}$ be a weakly subadditive, monotone and lower semi-continuous quantum relative entropy.
Assume that there exist two different $x_1,x_2\in\X$ such that 
for any non-commuting $\rho,\sigma\in\B(\hil)\pne$ with $\rho^0\le\sigma^0$, 
$\DU(\rho\|\sigma)<D^{q_{x_k}}(\rho\|\sigma)$, $k=1,2$. 
Assume that $(W_x)_{x\in\X}$ are invertible quantum states
such that the intersection of the commutants of $W_{x_1}$ and of $W_{x_2}$
is $\bC I_{\hil}$, and that $\sum_xP(x)\jsp{P}{W}(\logn W_x)\jsp{P}{W}$ is not a scalar multiple of the identity operator. Then
\begin{align*}
\limsup_{n\to+\infty}\frac{1}{n}D_P^{\bary,\qv}\bz W^{\otimes n}\jz< D_P^{\bary,\qv}(W).
\end{align*}
In particular, $D_P^{\bary,\qv}$ is not weakly additive.
\end{thm}

\section{Preliminaries}
\label{sec:prelim}

By a finite set $\X$, we will always mean a non-empty finite set.
For a finite set $\X$, we will denote by 
$\S(\X)$ the set of probability density functions on $\X$, i.e., 
$\S(\X):=\{P\in[0,1]^{\X}:\,\sum_xP(x)=1\}$.
For any natural number $d\in\bN$, we use the notation $[d]:=\{1,\ldots,d\}$.

Throughout the paper we use the convention
\begin{align*}
0\cdot(\pm\infty):=0.
\end{align*}

By a finite-dimensional Hilbert space we mean a finite-dimensional complex vector space $\hil$ equipped with an inner product. We follow the physics convention where the inner product is conjugate linear in its first variable and linear in the second. We will use the Dirac notation, where for 
$x,y\in\hil$, $\diad{y}{x}$ is a linear operator on $\hil$ acting as 
$\diad{y}{x}:\,z\mapsto y\inner{x}{z}$.

For a finite-dimensional Hilbert space $\hil$, let $\B(\hil)$ denote the set of all linear operators on $\hil$, 
and let $\B(\hil)_{\ge 0}$, $\B(\hil)_{\gneq 0}$, and 
$\B(\hil)_{>0}$ denote the set of 
positive semi-definite (PSD), non-zero positive semi-definite, and positive definite operators, respectively. 
$\S(\hil):=\{\rho\in\B(\hil)\p:\,\Tr\rho=1\}$ will denote the set of density operators, or states, on $\hil$.
For any $n\in\bN$, $\S_{\symm}(\hil^{\otimes n})$ will denote the set of permutation-invariant, or symmetric,
states on $\hil^{\otimes n}$.

For any non-empty set $\X$, let 
\begin{align*}
\B(\X,\hil),\ds\ds
\B(\X,\hil)\p,\ds\ds
\B(\X,\hil)\pne,\ds\ds
\B(\X,\hil)\pp,\ds\ds
\S(\X,\hil),
\end{align*}
denote the set of functions mapping from $\X$ into 
$\B(\hil)$,
$\B(\hil)\p$, $\B(\hil)\pne$, $\B(\hil)\pp$, and
$\S(\hil)$, respectively. 
We will normally use the notation $W=(W_x)_{x\in\X}$ to denote elements of 
$\B(\X,\hil)\p$.
We say that $W\in\B(\X,\hil)\p$ is \ki{classical} if 
there exists an orthonormal basis $(e_i)_{i\in\I}$ in $\hil$ such that 
$W_x=\sum_{i\in\I}\inner{e_i}{W_xe_i}\pr{e_i}$, $x\in\X$.

For a self-adjoint operator $A$, let $P^A_a:=\egy_{\{a\}}(A)$ denote the spectral projection of $A$
corresponding to the singleton $\{a\}\subset\bR$. 
(Here and henceforth $\egy_H$ stands for the characteristic (or indicator) function
of a set $H$.)
The projection onto the support of $A$ is $\sum_{a\ne 0}P^A_a$; in particular, if $A$ is 
positive semi-definite, it is equal to $\lim_{\alpha\searrow 0}A^{\alpha}=:A^0$. In general, 
we follow the convention that real powers of a positive semi-definite operator $A$ are taken 
only on its support, i.e., for any $x\in\bR$, $A^x:=\sum_{a>0}a^x P^A_a$.
In particular, $A\inv:=\sum_{a>0}a\inv P^A_a$ stands for the generalized inverse of $A$,
and $A\inv A=AA\inv=A^0$. 
For a projection $P$ on $\hil$, we will use the notation 
\begin{align*}
\S(P\hil):=\{\rho\in\S(\hil):\,\rho^0\le P\}
\end{align*}
for the set of states supported on $\ran P$.

For an operator $X\in\B(\hil)$, 
$\norm{X}_{\infty}:=\max\{\norm{X\psi}:\,\psi\in\hil,\,\norm{\psi}=1\}$
denotes the operator norm of $X$.

For a finite-dimensional Hilbert space $\hil$ and a natural number $n$, 
we denote by 
\begin{align*}
\povm(\hil,[n]):=\left\{M=(M_i)_{i=1}^{n}\in\B(\hil)\p^{[n]}:\,\sum_{i=1}^{n} M_i=I\right\}
\end{align*}
the set of \ki{$n$-outcome positive operator valued measures (POVMs)} on $\hil$.
Any $M\in\povm(\hil,[n])$ determines a CPTP map $\M:\,\B(\hil)\to\ell^{\infty}([n])$ by 
\begin{align*}
\M(\valt):=\sum_{i=1}^{n}(\Tr M_i(\valt))\egy_{\{i\}}.
\end{align*}

By $\log$ we denote the natural logarithm, and
we use two different extensions of it to $[0,+\infty]$, defined as
\begin{align*}
\log x:=\begin{cases}
-\infty,&x=0,\\
\log x,& x\in(0,+\infty),\\
+\infty,&x=+\infty,
\end{cases}
\ds\ds\ds
\nlog x:=\begin{cases}
0,&x=0,\\
\log x,& x\in(0,+\infty),\\
+\infty,&x=+\infty.
\end{cases}
\end{align*}
It is straightforward to verify that for two commuting PSD operators 
$\rho,\sigma\in\B(\hil)\pne$, 
\begin{align*}
\logn(\rho\sigma)=(\rho^0\wedge\sigma^0)(\logn\rho)(\rho^0\wedge\sigma^0)+(\rho^0\wedge\sigma^0)(\logn\sigma)(\rho^0\wedge\sigma^0),
\end{align*}
and for any two PSD operators $\rho_k\in\B(\hil_k)\pne$, $k=1,2$,
\begin{align}
\logn(\rho_1\otimes\rho_2)
&=
(\logn\rho_1)\otimes\rho_2^0+\rho_1^0\otimes\logn\rho_2\nn\\
&=
(\rho_1^0\otimes\rho_2^0)
\left[\logn\rho_1\otimes I_{\hil_2}+I_{\hil_1}\otimes \logn\rho_2\right](\rho_1^0\otimes\rho_2^0).
\label{eq:logn mult}
\end{align}

It is well known that the logarithm is operator monotone on positive definite operators,
while it is easy to see that $\logn$ is not operator monotone on PSD operators. However, we have the following:
\begin{lemma}\label{lemma:logn monotone}
For any  $A,B\in\B(\hil)\p$,
\begin{align*}
A\le B\ds\ds\imp\ds\ds \logn A\le A^0(\logn B)A^0.
\end{align*} 
\end{lemma}
\begin{proof}
Assume that $A\le B$. Then $A^0\le B^0$, and thus $B^0$ acts as the identity on the support of $A$. By the operator monotonicity of $\log$, we have
\begin{align*}
\logn(A+\ep B^0)\le\logn(B+\ep B^0),\ds\ds\ds\ep>0,
\end{align*}
whence
\begin{align*}
A^0\logn(A+\ep B^0)A^0\le A^0\logn(B+\ep B^0)A^0,\ds\ds\ds\ep>0.
\end{align*}
Taking the limit $\ep\searrow 0$ yields $\logn A\le A^0(\logn B)A^0$.
\end{proof}

The following is a strict version of the Hansen-Pedersen characterization of the operator concavity of the 
matrix logarithm \cite{hansen2003jensen}. We give a proof for completeness.
\begin{lemma}\label{lemma:logn opJensen}
Let $A\in\B(\hil)\pne$ and $P\in\B(\hil)$ be a projection such that $P\le A^0$. Then 
\begin{align}\label{eq:logn opJensen}
P\bz\logn(PAP)\jz P\ge P(\logn A)P,
\end{align}
with equality if and only if $PA=AP$.
\end{lemma}
\begin{proof}
By the operator Jensen inequality for the log, 
\begin{align*}
\logn\bz PAP+\ep(A^0-P)\jz\ge P(\logn A)P+(A^0-P)(\log\ep).
\end{align*}
Multiplying by $P$ from both sides and taking $\ep\searrow 0$ yields the inequality in \eqref{eq:logn opJensen}.

If $PA=AP$ then it is obvious that \eqref{eq:logn opJensen} holds with equality. 
Conversely, assume that \eqref{eq:logn opJensen} holds with equality.
Let $A=\sum_{a>0}aP^A_a$ be the spectral decomposition of $A$, and let 
$PAP=\sum_{k=1}^r\lambda_k\pr{\psi_k}$ be an eigen-decomposition of $PAP$, with all $\lambda_k>0$. 
Then, for any $k\in[r]$, 
\begin{align*}
\sum_{a>0}(\log a)\norm{P^A_a\psi_k}^2
&=
\inner{\psi_k}{P(\logn A)P\psi_k}
=
\inner{\psi_k}{P\bz\logn(PAP)\jz P\psi_k}\\
&=
\log\lambda_k
=
\log\inner{\psi_k}{PAP\psi_k}
=
\log\bz\sum_{a>0}a\norm{P^A_a\psi_k}^2\jz,
\end{align*}
where the second equality follows from the assumed equality in \eqref{eq:logn opJensen}, 
and we used that $P\psi_k=\psi_k$. The strict concavity of the scalar logarithm then yields
that for every $k\in[r]$ there exists a unique $a_k$ such that $P^A_{a_k}\psi_k=\psi_k$
(and hence $P^A_a\psi_k=0$ for every $a\ne a_k$). Using then $P=\sum_{k=1}^r\pr{\psi_k}$, we get that 
for any $a\in(0,+\infty)$, 
$P^A_aP=\sum_{k:\,a_k=a}\pr{\psi_k}=PP^A_a$, whence $AP=PA$.
\end{proof}

For two positive semidefinite operators $\rho,\sigma\in\B(\hil)\p$, let 
\begin{align*}
\acc{\sigma}{\rho}:=\max\left\{X\in\B(\hil)\p:\,X\le\sigma,\,X^0\le\rho^0\right\}
\end{align*}
be the \ki{absolutely continuous part} of $\sigma$ with respect to $\rho$ \cite{Anderson1971,Ando_Lebesgue}.
It is a non-trivial fact that the above maximum exists; moreover, it can be given explicitly as 
\begin{align}
\acc{\sigma}{\rho}&=
P\sigma P-P\sigma (P^{\perp}\sigma P^{\perp})\inv\sigma P\label{eq:ac explicit}\\
&=
((\rho^0\wedge\sigma^0)\sigma\inv(\rho^0\wedge\sigma^0))\inv,\label{eq:ac explicit2}
\end{align}
where $P$ is any projection satisfying $\rho^0\wedge\sigma^0\le P\le\rho^0$.
For the expression in \eqref{eq:ac explicit}, see \cite{Anderson1971}.
The expression in \eqref{eq:ac explicit2} can be obtained from \eqref{eq:ac explicit} using the well-known formula for the inverse of a $2\times 2$ block matrix;
see also \cite[page 436]{Krein1947SelfAdjointExtensionsI}, or
\cite{MasseyStojanoff2004}. We give an alternative self-contained proof of it in 
Appendix \ref{sec:ac}.


It is obvious by definition that for $\rho_1,\rho_2,\sigma\in\B(\hil)\p$, 
\begin{align}\label{eq:acc monotone}
\rho_1\le\rho_2\ds\imp\ds \acc{\sigma}{\rho_1}\le \acc{\sigma}{\rho_2}.
\end{align}

For any $\gamma\in(0,1)$, the \ki{$\gamma$-weighted Kubo-Ando geometric mean} of 
two PSD operators $\rho,\sigma\in\B(\hil)\p$ is defined as \cite{KA}
\begin{align*}
\sigma\#_{\gamma}\rho
&:=
\lim_{\ep\searrow 0}(\rho+\ep I)^{1/2}\bz(\rho+\ep I)^{-1/2}(\sigma+\ep I)(\rho+\ep I)^{-1/2}\jz^{1-\gamma}(\rho+\ep I)^{1/2}\\
&=
\rho^{1/2}\bz\rho^{-1/2}\acc{\sigma}{\rho}\rho^{-1/2}\jz^{1-\gamma}\rho^{1/2},
\end{align*}
where the equality is due to \cite{Kosaki_ac}.
We define
\begin{align}\label{eq:KA limit at 0}
\sigma\#_0\rho:=\lim_{\gamma\searrow 0}\sigma\#_{\gamma}\rho=\acc{\sigma}{\rho}.
\end{align}

The absolutely continuous part is tensor multiplicative, i.e., for any $\rho_k,\sigma_k\in\B(\hil^{(k)})\p$, $k=1,2$, 
\begin{align}\label{eq:acc mult}
\acc{\sigma_1\otimes\sigma_2}{\rho_1\otimes\rho_2}
=
\acc{\sigma_1}{\rho_1}\otimes\acc{\sigma_2}{\rho_2};
\end{align}
see, e.g., \cite[Proposition A.7]{mosonyi2022geometric} for a direct proof.
Alternatively, it follows from the easily verifiable tensor multiplicativity 
$(\sigma_1\otimes\sigma_2)\#_{\gamma}(\rho_1\otimes\rho_2)
=(\sigma_1\#_{\gamma}\rho_1)\otimes(\sigma_2\#_{\gamma}\rho_2)$ and \eqref{eq:KA limit at 0},
or from \eqref{eq:ac explicit2}; see Corollary \ref{cor:ac mult}.

We will use the following variant of Fekete's lemma.

\begin{lemma}\label{lemma:subadditive Fekete}
For any sequence $a_n\in\bR\cup\{\pm\infty\}$, $n\in\bN$, 
\begin{align}\label{eq:subadditive Fekete}
\inf_{n\in\bN}\frac{1}{n}a_n\le
\liminf_{n\to+\infty}\frac{1}{n}a_n
\le
\limsup_{n\to+\infty}\frac{1}{n}a_n
\le
\sup_{n\in\bN}\frac{1}{n}a_n, 
\end{align}
and the following hold:
\begin{enumerate}
\item\label{item:Fekete1}
If the sequence is weakly subadditive, i.e., for every $k,n\in\bN$, 
$a_{kn}\le ka_n$, then the first inequality in \eqref{eq:subadditive Fekete} is an equality.
\item\label{item:Fekete2}
If $a_n<+\infty$ for every $n\in\bN$, and the sequence is subadditive, i.e., for every $m,n\in\bN$, 
$a_{m+n}\le a_m+a_n$, then the first and the second inequalities in \eqref{eq:subadditive Fekete} 
are equalities.
\item\label{item:Fekete3}
If the sequence is weakly superadditive, i.e., for every $k,n\in\bN$, 
$a_{kn}\ge ka_n$, then the last inequality in \eqref{eq:subadditive Fekete} is an equality.
\item\label{item:Fekete4}
If $a_n>-\infty$ for every $n\in\bN$, and the sequence is superadditive, i.e., for every $m,n\in\bN$, 
$a_{m+n}\ge a_m+a_n$, then the second and the last inequalities in \eqref{eq:subadditive Fekete} 
are equalities.

\end{enumerate}
\end{lemma}
\begin{proof}
The inequalities in \eqref{eq:subadditive Fekete} are obvious, and the assertions in 
\ref{item:Fekete2} and \ref{item:Fekete4} are the usual Fekete's lemma.

Assume now that the sequence is weakly subadditive, and let
$\kappa:=\inf_{n\in\bN}a_n/n$. If $\kappa=+\infty$ then the first inequality in \eqref{eq:subadditive Fekete} is trivially an equality. Assume therefore that $\kappa<+\infty$.
Then for any $\kappa'>\kappa$, there exists an $n_{\kappa'}\in\bN$ such that 
$a_{n_{\kappa'}}/n_{\kappa'}<\kappa'$, whence
\begin{align*}
\liminf_{n\to+\infty}\frac{1}{n}a_n
\le
\liminf_{k\to+\infty}\frac{1}{kn_{\kappa'}}a_{kn_{\kappa'}}
\le
\frac{1}{n_{\kappa'}}a_{n_{\kappa'}}
<
\kappa',
\end{align*}
where the second inequality follows from the weak subadditivity property, and the rest are trivial.
Since this holds for any $\kappa'>\kappa$, we get that the first inequality in 
\eqref{eq:subadditive Fekete} is an equality.
This proves the assertion in \ref{item:Fekete1}, and the assertion in \ref{item:Fekete3} follows immediately by applying \ref{item:Fekete1} to $\tilde a_n:=-a_n$, $n\in\bN$.
\end{proof}

\section{Regularized barycentric R\'enyi divergences}

\subsection{Quantum divergences}
\label{sec:qdiv}

Let $\X$ be a non-empty finite set. By an \ki{$\X$-variable quantum divergence} $\divv$
we mean a function on collections of PSD matrices
\begin{align*}
\dom(\divv)\subseteq\cup_{d\in\bN}\,\B(\X,\bC^d)\p,
\ds\ds
\divv:\,\dom(\divv)
\to\bR\cup\{\pm\infty\},
\end{align*}
that is \ki{invariant under isometries}, i.e., if $V:\,\bC^{d_1}\to\bC^{d_2}$ is an isometry then 
\begin{align*}
&V\bz\dom(\divv)\cap\B(\X,\bC^{d_1})\jz V^*\subseteq\dom(\divv),\ds\text{and}\ds\\
&\divv\bz V\W V^*\jz
=
\divv\bz \W\jz,\ds\ds\ds \W\in\dom(\divv)\cap\B(\X,\bC^{d_1}),
\end{align*}
where $VWV^*:=(VW_xV^*)_{x\in\X}$.
Due to the isometric invariance, $\divv$ may be extended to collections of operators
on any finite-dimensional Hilbert space $\hil$, by 
\begin{align*}
&\dom_{\hil}(\divv):=\left\{W\in\B(\X,\hil)\p:\,\exists\,V:\,\hil\to\bC^d\text{ isometry: }VWV^*\in\dom(\divv)\right\},\\
&\divv(\W):=\divv(V\W V^*),\ds\ds\ds \W\in\dom_{\hil}(\divv),
\end{align*}
where $V$ in the second line is any isometry $V:\,\hil\to\bC^d$ such that 
$VWV^*\in\dom(\divv)$. The isometric invariance of $\divv$ on $\dom(\divv)$ 
guarantees that this extension is well-defined, in the sense
that the value of $\divv(V\W V^*)$
is independent of the choice of $d$ and $V$. 
Clearly, this extension is again invariant under isometries, i.e.,
for any $\W\in\dom_{\hil}(\divv)$ and 
$V:\,\hil\to\kil$ isometry,
$VWV^*\in\dom_{\kil}(\divv)$, and
$\divv(V\W V^*)=\divv(W)$.

A quantum divergence $\divv$ is called
\begin{itemize}
\item
\ki{monotone}, if for any $W\in\dom_{\hil}(\divv)$ and any CPTP (completely positive and trace-preserving) map
$\map:\,\B(\hil)\to\B(\kil)$, $\map(W):=(\map(W_x))_{x\in\X}\in\dom_{\kil}(\divv)$, and 
$\divv(\map(W))\le\divv(W)$;
\item
\ki{weakly subadditive}, if for any $W\in\dom_{\hil}(\divv)$, 
$W^{\otimes n}:=(W_x^{\otimes n})_{x\in\X}\in\dom_{\hil^{\otimes n}}(\divv)$, and 
for any $W\in\dom_{\hil}(\divv)$ and any $n\in\bN$,
\begin{align}\label{eq:weak subadd def}
\divv(W^{\otimes n})\le n\divv(W);
\end{align}
\item
\ki{weakly additive}, if the above holds, but with equality in \eqref{eq:weak subadd def};
\item
\ki{subadditive}, if for any $W^{(k)}\in\dom_{\hil^{(k)}}(\divv)$, $k=1,2$, 
$W^{(1)}\otimes W^{(2)}:=(W^{(1)}_x\otimes W^{(2)}_x)_{x\in\X}\in\dom_{\hil^{(1)}\otimes\hil^{(2)}}(\divv)$, and 
\begin{align}\label{eq:subadd def}
\divv(W^{(1)}\otimes W^{(2)})\le\divv(W^{(1)})+\divv(W^{(2)}),
\end{align}
whenever the RHS is well-defined;
\item
\ki{additive}, if the above holds, but with equality in \eqref{eq:subadd def}.
\end{itemize}
\ki{Weak superadditivity} and \ki{superadditivity} of $\divv$ are defined analogously to the above, with  opposite inequality signs in \eqref{eq:weak subadd def} and in \eqref{eq:subadd def}, respectively.

For any quantum divergence $\divv$ satisfying 
$W\in\dom_{\hil}(\divv)$ $\imp$ $W^{\otimes n}\in\dom_{\hil^{\otimes n}}(\divv)$, 
we define its \ki{lower and upper regularized versions} 
$\lreg{\divv}$ and $\ureg{\divv}$, respectively, as
\begin{align*}
\lreg{\divv}(W)&:=\liminf_{n\to+\infty}\frac{1}{n}\divv(W^{\otimes n}),\\
\ureg{\divv}(W)&:=\limsup_{n\to+\infty}\frac{1}{n}\divv(W^{\otimes n}),\ds\ds\ds W\in\dom_{\hil}(\divv).
\end{align*}
If $\lreg{\divv}(W)=\ureg{\divv}(W)$ for every $W\in\dom_{\hil}(\divv)$ then we use the notation
\begin{align*}
\breg{\divv}(W)&:=\lim_{n\to+\infty}\frac{1}{n}\divv(W^{\otimes n}),
\ds\ds\ds W\in\dom_{\hil}(\divv).
\end{align*}

\begin{lemma}\label{lemma:regularized div liminf}
(i) For any weakly subadditive quantum divergence $\divv$, its regularized version is weakly additive, and for any 
$W\in\dom_{\hil}(\divv)$,
\begin{align}\label{eq:regularized div liminf}
\lreg{\divv}(W)
=
\inf_{n\in\bN}\frac{1}{n}\divv(W^{\otimes n}).
\end{align}
If, moreover, $\divv$ is subadditive and $\divv(W)<+\infty$ then we further have
$\lreg{\divv}(W)=\ureg{\divv}(W)$.
\smallskip

(ii) For any weakly superadditive quantum divergence $\divv$, its regularized version is weakly additive, and for any 
$W\in\dom_{\hil}(\divv)$,
\begin{align}\label{eq:regularized div limsup}
\ureg{\divv}(W)
=
\sup_{n\in\bN}\frac{1}{n}\divv(W^{\otimes n}).
\end{align}
If, moreover, $\divv$ is superadditive and $\divv(W)>-\infty$ then we further have
$\lreg{\divv}(W)=\ureg{\divv}(W)$.
\end{lemma}
\begin{proof}
The equalities in \eqref{eq:regularized div liminf} and \eqref{eq:regularized div limsup} and the assertions 
about subadditive/superadditive divergences are immediate from Lemma \ref{lemma:subadditive Fekete}, and hence we only need to prove the weak additivity claims. 

Assume that $\divv$ is weakly subadditive.
For any $W\in\dom_{\hil}(\divv)$ and any $k\in\bN$,
\begin{align*}
\lreg{\divv}(W^{\otimes k})
=
\liminf_{n\to+\infty}\frac{1}{n}\divv(W^{\otimes kn})
=k\liminf_{n\to+\infty}\frac{1}{kn}\divv(W^{\otimes kn})\ge k\lreg{\divv}(W),
\end{align*}
simply by definition,
while weak subadditivity yields
\begin{align*}
\lreg{\divv}(W^{\otimes k})
=
\liminf_{n\to+\infty}\frac{1}{n}\divv(W^{\otimes kn})
\le
\liminf_{n\to+\infty}\frac{1}{n}k\divv(W^{\otimes n})
=k\lreg{\divv}(W).
\end{align*}
This proves weak additivity. 
\end{proof}

\begin{definition}
For two divergences $\divv$ and $\divv'$, we use the notation 
$\divv\le\divv'$ if $\dom(\divv)=\dom(\divv')$ and for any $W\in\dom(\divv)$, 
$\divv(W)\le\divv'(W)$. 
\end{definition}

\subsection{Quantum relative entropies}
\label{sec:relentr}

The \ki{classical relative entropy} or \ki{Kullback-Leibler divergence} of two non-zero non-negative 
functions $\rho,\sigma$ on a finite set $\I$ is defined as
\begin{align*}
D^{\cl}(\rho\|\sigma):=
\begin{cases}
\sum_{i\in\I}\rho(i)\left[\logn\rho(i)-\logn\sigma(i)\right],&\rho\ll\sigma,\\
+\infty,&\text{otherwise},
\end{cases}
\end{align*}
where $\rho\ll\sigma$ means $\sigma(i)=0\imp\rho(i)=0$.
Moreover, we define 
$D^{\cl}(\rho\|0):=+\infty$ for all $\rho\gneq 0$.

By a \ki{quantum relative entropy} we mean a function
\begin{align*}
D^q:\,\medcup_{d\in\bN}\bz\B(\bC^d)\pne\times\B(\bC^d)\p\jz\to\bR\cup\{+\infty\}
\end{align*}
satisfying the following four properties:
\begin{itemize}
\item
\ki{Isometric invariance:} For any isometry $V:\,\bC^d\to\bC^{d'}$ and any 
$\rho,\sigma\in\B(\bC^d)\p$ with $\rho\ne 0$, 
$D^q(V\rho V^*\|V\sigma V^*)=D^q(\rho\|\sigma)$.
\item
\ki{Classical reduction:} For any two PSD operators $\rho=\sum_{k=1}^d\tilde\rho(k)\pr{\egy_{\{k\}}}$, 
$\sigma=\sum_{k=1}^d\tilde\sigma(k)\pr{\egy_{\{k\}}}$ with $\rho\ne 0$ that are diagonal in the canonical orthonormal basis 
$(\egy_{\{k\}})_{k=1}^d$ of $\bC^d$,
$D^q(\rho\|\sigma)=D^{\cl}\bz\bz\tilde\rho(k)\jz_{k=1}^d\|\bz\tilde\sigma(k)\jz_{k=1}^d\jz$.
\item
\ki{Positivity:} For any two states $\rho,\sigma\in\S(\bC^d)$, 
$D^q(\rho\|\sigma)\ge 0$ with equality if and only if $\rho=\sigma$. 
\item
\ki{Finiteness:} $D^q(\rho\|\sigma)=+\infty$ $\iff$ $\rho^0\nleq\sigma^0$.
\end{itemize}
Due to the isometric invariance, any quantum relative entropy $D^q$ 
is a $2=\{0,1\}$-variable quantum divergence as defined in Section \ref{sec:qdiv}.
In particular, it can be uniquely extended to pairs of 
PSD operators on any finite-dimensional Hilbert space $\hil$ by defining
$D^q(\rho\|\sigma):=D^q(V\rho V^*\|V\sigma V^*)$ for any $\rho,\sigma\in\B(\hil)\p$
with $\rho\ne 0$, where $V$ is an arbitrary unitary from $\hil$ to $\bC^{\dim\hil}$. The resulting quantity is invariant under isometries between arbitrary finite-dimensional Hilbert spaces, it has the reduction property in the sense that if 
$\rho,\sigma\in\B(\hil)\pne$ are both diagonal in some orthonormal basis $(e_i)_{i\in\I}$ then 
$D^q(\rho\|\sigma)=D^{\cl}((\inner{e_i}{\rho e_i})_{i\in\I}\|(\inner{e_i}{\sigma e_i})_{i\in\I})$,
and it satisfies the obvious reformulations of the positivity and the finiteness conditions.

For any quantum relative entropy $D^q$, its \ki{normalized version} $D_1^q$ is defined 
for any $\rho,\sigma\in\B(\hil)\p$ with $\rho\ne 0$ as
\begin{align*}
D_1^q(\rho\|\sigma):=
\frac{1}{\Tr\rho}D^q(\rho\|\sigma).
\end{align*}

The key properties of a quantum relative entropy $D^q$ are defined via the properties of $D_1^q$ as a $2$-variable quantum divergence. In particular,
a quantum relative entropy $D^q$ is called 
\begin{itemize}
\item
\ki{monotone} if 
$D_1^q$ is monotone, or equivalently,
for any $\rho,\sigma\in\B(\hil)\pne$ and CPTP 
(completely positive and trace-preserving) map 
$\map:\,\B(\hil)\to\B(\kil)$,
$D^q(\map(\rho)\|\map(\sigma))\le D^q(\rho\|\sigma)$;
\item
\ki{weakly subadditive} if $D_1^q$ is weakly subadditive, or equivalently,
for any $\rho,\sigma\in\B(\hil)\p$ with $\rho\ne 0$, and any $n\in\bN$, 
$D^q(\rho^{\otimes n}\|\sigma^{\otimes n})\le n(\Tr\rho)^{n-1}D^q(\rho\|\sigma)$;
\item
\ki{weakly additive} if $D_1^q$ is weakly additive, or equivalently,
for any $\rho,\sigma\in\B(\hil)\p$ with $\rho\ne 0$, and any $n\in\bN$, 
$D^q(\rho^{\otimes n}\|\sigma^{\otimes n})=n(\Tr\rho)^{n-1}D^q(\rho\|\sigma)$;
\item
\ki{additive} if $D_1^q$ is additive, or equivalently, for any $\rho_1,\sigma_1\in\B(\hil_1)\pne$ and 
$\rho_2,\sigma_2\in\B(\hil_2)\pne$,
$D^q(\rho_1\otimes\rho_2\|\sigma_1\otimes\sigma_2)=
(\Tr\rho_2)D^q(\rho_1\|\sigma_1)+(\Tr\rho_1)D^q(\rho_2\|\sigma_2)$;
\end{itemize}

Among the monotone and weakly additive quantum relative entropies there is a smallest one
\cite{HP}, the 
\ki{Umegaki relative entropy} \cite{Umegaki}, defined for $\rho,\sigma\in\B(\hil)\p$,
$\rho\ne 0$, as
\begin{align*}
\DU(\rho\|\sigma):=\begin{cases}
\Tr\rho(\logn\rho-\logn\sigma),&\rho^0\le\sigma^0,\\
+\infty,&\text{otherwise},
\end{cases}
\end{align*}
and a largest one \cite{Matsumoto_newfdiv}, the \ki{Belavkin-Staszewski relative entropy} \cite{BS}, 
defined for $\rho,\sigma\in\B(\hil)\p$ with $\rho\ne 0$ as
\begin{align*}
\DBS(\rho\|\sigma):=\begin{cases}
\Tr\rho\logn\bz\rho^{1/2}\sigma\inv\rho^{1/2}\jz,&\rho^0\le\sigma^0,\\
+\infty,&\text{otherwise}.
\end{cases}
\end{align*}
We will also use the \ki{measured relative entropy} \cite{Donald1986}, defined 
for any $\rho,\sigma\in\B(\hil)\p$ with $\rho\ne 0$ as
\begin{align*}
D^{\meas}(\rho\|\sigma):=\sup\{D^{\cl}\bz\M(\rho)\|\M(\sigma)\jz:\,M\in\povm(\hil,[n]),\,n\in\bN\},
\end{align*}
which is minimal among all monotone quantum relative entropies. 

We summarize the above extremality properties of the above quantum relative entropies with a slightly stronger statement on the extremality of $\DU$ and $\DBS$, as follows:

\begin{lemma}\label{lemma:sandwich}
$D^{\meas},\DU$ and $\DBS$ are monotone quantum relative entropies, and 
$\DU$ and $\DBS$ are also additive.  
Any monotone quantum relative entropy $D^q$ satisfies
\begin{align}\label{eq:monotone relentr sandwich}
D^{\meas}\le D^q\le D^{\max},
\end{align}
and if $D^q$ is also weakly subadditive then 
\begin{align*}
\DU\le D^q\le D^{\max}.
\end{align*}
\end{lemma}
\begin{proof}
The properties stated for $D^{\meas}$, $\DU$ and $\DBS$ are well known, as is 
\eqref{eq:monotone relentr sandwich}; we refer to 
\cite{BS,HP,Matsumoto_newfdiv} for a detailed discussion.
Assume now that $D^q$ is monotone and weakly subadditive. Then for any 
$\rho,\sigma\in\B(\hil)\p$ with $\rho\ne 0$ and any $n\in\bN$,
\begin{align}\label{eq:regrelentr proof1}
D_1^{\meas}\bz\rho^{\otimes n}\|\sigma^{\otimes n}\jz
\le 
D_1^q\bz\rho^{\otimes n}\|\sigma^{\otimes n}\jz
\le
D^{\max}_1\bz\rho^{\otimes n}\|\sigma^{\otimes n}\jz,
\end{align}
according to \eqref{eq:monotone relentr sandwich}.
Thus,
\begin{align*}
D^{\Um}_1(\rho\|\sigma)
&=
\lim_{n\to+\infty}\frac{1}{n}D^{\meas}_1(\rho^{\otimes n}\|\sigma^{\otimes n})
=
\liminf_{n\to+\infty}\frac{1}{n}D^{\meas}_1(\rho^{\otimes n}\|\sigma^{\otimes n})\\
&\le
\liminf_{n\to+\infty}\frac{1}{n}D^q_1(\rho^{\otimes n}\|\sigma^{\otimes n})
=
\inf_{n\in\bN}\frac{1}{n}D^q_1(\rho^{\otimes n}\|\sigma^{\otimes n})
\le
D^q_1(\rho\|\sigma)\\
&\le
D^{\max}_1(\rho\|\sigma),
\end{align*}
where the first equality is due to \cite{HP}, the second equality is trivial, 
the first inequality follows from \eqref{eq:regrelentr proof1},
the third equality follows from Lemma \ref{lemma:regularized div liminf}
due to the assumption that $D^q$ is weakly subadditive,
the second inequality is trivial, and the last inequality follows from 
\eqref{eq:monotone relentr sandwich}.
Multiplying over by $\Tr\rho$ yields $\DU(\rho\|\sigma)\le D^q(\rho\|\sigma)\le\DBS(\rho\|\sigma)$, as required.
\end{proof}


Various other notions of additive and monotone quantum relative entropies were defined e.g., in 
\cite{mosonyi2022geometric} and \cite{CapelRico_fdiv} by considering different interpolations between the 
Umegaki and the Belavkin-Staszewski relative entropies.
In particular, for any $\gamma\in(0,1)$, the \ki{$\gamma$-weighted geometric Umegaki relative entropy}
$D^{\Um,\#_{\gamma}}$ was defined in \cite{mosonyi2022geometric} for arbitrary 
$\rho,\sigma\in\B(\hil)\p$ with $\rho\ne 0$ as
\begin{align*}
D^{\Um,\#_{\gamma}}(\rho\|\sigma)
&:=\frac{1}{1-\gamma}\DU(\rho\|\sigma\#_{\gamma}\rho)\\
&=
\begin{cases}
\frac{1}{1-\gamma}\Tr\rho\left[\logn\rho-\logn\bz\rho^{1/2}\bz\rho^{-1/2}\acc{\sigma}{\rho}\rho^{-1/2}\jz^{1-\gamma}\rho^{1/2}\jz\right],
&\rho^0\le\sigma^0,\\
+\infty,&\text{otherwise}.
\end{cases}
\end{align*}
The above definition can be extended to $\gamma=0$, yielding
\begin{align*}
D^{\Um,\#_{0}}(\rho\|\sigma):=\DU(\rho\|\sigma\#_{0}\rho)
=
\left\{
\begin{array}{ll}
\Tr\rho\left[\logn\rho-\logn\bz\acc{\sigma}{\rho}\jz\right],
&\rho^0\le\sigma^0,\\
+\infty,&\text{otherwise}.
\end{array}\right\}
=
\lim_{\gamma\searrow 0}D^{\Um,\#_{\gamma}}(\rho\|\sigma).
\end{align*}
As it was shown in \cite{mosonyi2022geometric}, $D^{\Um,\#_{\gamma}}$ is monotone and additive for every 
$\gamma\in(0,1)$, and hence so is $D^{\Um,\#_{0}}$, too.
As a consequence, 
\begin{align}\label{eq:DU vs sharpgamma}
\DU\le D^{\Um,\#_{\gamma}}\le\DBS,\ds\ds\ds\gamma\in[0,1).
\end{align}
We define
\begin{align*}
D^{\Um,\#_1}:=\DBS,
\end{align*}
justified by Proposition \ref{prop:sharp limit 1} below.

We will need the following lemma in the next section.

\begin{lemma}\label{lemma:sharp0 bound}
Let $\rho,\sigma,S\in\B(\hil)\pne$. Then
\begin{align*}
\rho^0\le S^0\ds\imp\ds 
D^{\Um,\#_\gamma}\bz\rho\|\sigma\jz
=
D^{\Um,\#_\gamma}\bz\rho\|\acc{\sigma}{S}\jz
\ge
\DU\bz\rho\|\acc{\sigma}{S}\jz,\ds\ds\ds\gamma\in[0,1].
\end{align*}
\end{lemma}
\begin{proof}
If $\rho^0\nleq\sigma^0$ then $\rho^0\nleq(\acc{\sigma}{S})^0$, whence
$D^{\Um,\#_\gamma}\bz\rho\|\sigma\jz=+\infty=D^{\Um,\#_\gamma}\bz\rho\|\acc{\sigma}{S}\jz$, 
$\gamma\in[0,1]$,
and hence the statement holds trivially.
Hence, for the rest we assume that $\rho^0\le\sigma^0$.

Note that 
\begin{align*}
\acc{\sigma}{\rho}=\acc{\bz\acc{\sigma}{S}\jz}{\rho},
\end{align*}
whence for every $\gamma\in[0,1)$,
\begin{align*}
\sigma\#_{\gamma}\rho
=
\rho^{1/2}\bz\rho^{-1/2}\acc{\sigma}{\rho}\rho^{-1/2}\jz^{1-\gamma}\rho^{1/2}
=
\rho^{1/2}\bz\rho^{-1/2}\acc{\bz\acc{\sigma}{S}\jz}{\rho}\rho^{-1/2}\jz^{1-\gamma}\rho^{1/2}
=
(\acc{\sigma}{S})\#_{\gamma}\rho.
\end{align*}
This proves
$D^{\Um,\#_\gamma}\bz\rho\|\sigma\jz
=
D^{\Um,\#_\gamma}\bz\rho\|\acc{\sigma}{S}\jz$
for every $\gamma\in[0,1)$. 
(This argument also works without assuming $\rho^0\le\sigma^0$.)

For $\gamma=1$, note that by \eqref{eq:ac explicit2}, $\acc{\sigma}{S}=((\sigma^0\wedge S^0)\sigma\inv (\sigma^0\wedge S^0))\inv$, whence
\begin{align*}
\rho^{1/2}(\acc{\sigma}{S})\inv\rho^{1/2}
=
\underbrace{\rho^{1/2}(\sigma^0\wedge S^0)}_{=\rho^{1/2}}\sigma\inv \underbrace{(\sigma^0\wedge S^0)\rho^{1/2}}_{=\rho^{1/2}}
=
\rho^{1/2}\sigma\inv\rho^{1/2}.
\end{align*}
Hence, 
\begin{align*}
\DBS(\rho\|\acc{\sigma}{S})
=
\Tr\rho\logn\bz\rho^{1/2}(\acc{\sigma}{S})\inv\rho^{1/2}\jz
=
\Tr\rho\logn\bz\rho^{1/2}\sigma\inv\rho^{1/2}\jz
=
\DBS(\rho\|\sigma).
\end{align*}

Finally, the inequalities
$D^{\Um,\#_\gamma}\bz\rho\|\acc{\sigma}{S}\jz
\ge
\DU\bz\rho\|\acc{\sigma}{S}\jz$, $\gamma\in[0,1]$,
follow by the minimality of $\DU$ and the fact that all $D^{\Um,\#_{\gamma}}$ are monotone and additive.
\end{proof}

\begin{rem}
The inequality 
$D^{\Um,\#_0}\bz\rho\|\sigma\jz\ge\DU\bz\rho\|\acc{\sigma}{S}\jz$ can be proved in an alternative way as 
follows. When $\rho^0\nleq\sigma^0$, the inequality holds trivially, and hence for the rest we assume that 
$\rho^0\le\sigma^0$.
According to \eqref{eq:acc monotone}, $\rho^0\le S^0$ implies that $\acc{\sigma}{\rho}\le\acc{\sigma}{S}$, and hence, by Lemma \ref{lemma:logn monotone}, 
\begin{align*}
\rho^0(\logn\acc{\sigma}{\rho})\rho^0
\le
\rho^0(\logn\acc{\sigma}{S})\rho^0.
\end{align*}
Thus,
\begin{align*}
D^{\Um,\#_0}\bz\rho\|\sigma\jz
=
\Tr\rho\logn\rho-\Tr\rho\logn\acc{\sigma}{\rho}
\ge
\Tr\rho\logn\rho-\Tr\rho\logn\acc{\sigma}{S}
=
\DU\bz\rho\|\acc{\sigma}{S}\jz.
\end{align*}
\end{rem}

\begin{prop}\label{prop:sharp limit 1}
For any $\rho,\sigma\in\B(\hil)\p$ with $\rho\ne 0$, 
\begin{align*}
D^{\Um,\#_{\gamma}}(\rho\|\sigma)\nearrow \DBS(\rho\|\sigma)\ds\ds\text{as}\ds\ds\gamma\nearrow 1.
\end{align*}
\end{prop}
\begin{proof}
The case where $\rho^0=\sigma^0$ follows from \cite[Proposition 4.14]{mosonyi2022geometric}.
The case $\rho^0\le\sigma^0$ follows from this and Lemma \ref{lemma:sharp0 bound} by noting that 
$\rho^0\le\sigma^0$ implies
$\rho^0=(\acc{\sigma}{\rho})^0$.
The case $\rho^0\nleq\sigma^0$ is trivial.
\end{proof}

It was shown in \cite[Proposition 4.14]{mosonyi2022geometric} that for any $\rho,\sigma\in\B(\hil)\pne$, 
$(0,1]\ni\gamma\mapsto D^{\Um,\#_{\gamma}}(\rho\|\sigma)$ is monotone increasing, which we already used in Proposition \ref{prop:sharp limit 1} above. In fact, a stronger strict monotonicity also holds as follows.

\begin{prop}\label{prop:gamma-relentr mon}
Let $\rho,\sigma\in\B(\hil)\pne$ be such that $\rho^0\le\sigma^0$, and let $f(\gamma):=D^{\Um,\#_{\gamma}}(\rho\|\sigma)$, $\gamma\in[0,1]$.
We have the following.
\begin{enumerate}
\item\label{item:gamma-relentr mon1}
Either $f$ is strictly increasing on $[0,1]$ or it is constant.
\item\label{item:gamma-relentr mon2}
If $\rho^0=\sigma^0$ then $f$ is constant if and only if $\rho$ and $\sigma$ commute.
\item\label{item:gamma-relentr mon3}
If $\rho$ and $\sigma$ do not commute then $\DU(\rho\|\sigma)<D^{\Um,\#_{\gamma}}(\rho\|\sigma)$ for every 
$\gamma\in(0,1]$.
\item\label{item:gamma-relentr mon4}
$\DU(\rho\|\sigma)\le D^{\Um,\#_0}(\rho\|\sigma)$ with equality if and only if $\sigma$ commutes with $\rho^0$.
\end{enumerate}
\end{prop}
\begin{proof}
By definition, 
\begin{align*}
f(\gamma)=\frac{1}{1-\gamma}\Tr\rho\logn\rho-\frac{1}{1-\gamma}\Tr\rho\logn\bz
\rho^{1/2}\exp\bz(1-\gamma)\logn\bz\rho^{-1/2}\acc{\sigma}{\rho}\rho^{-1/2}\jz\jz\rho^{1/2}\jz,
\ds\ds\ds\gamma\in[0,1).
\end{align*}
It is easy to see that $f$ is continuous on $[0,1]$ and real analytic on $(0,1)$, and we know from 
\cite[Proposition 4.14]{mosonyi2022geometric} that it is monotone increasing. Thus, if there 
exist $0\le\gamma_1<\gamma_2\le 1$ such that $f(\gamma_1)=f(\gamma_2)$ then, by the identity theorem for real 
analytic functions, $f$ is constant on $(0,1)$, and therefore also on $[0,1]$. 
This proves \ref{item:gamma-relentr mon1}.

It is clear that if $\rho$ and $\sigma$ commute then $f$ is constant on $[0,1]$. Assume that 
$\rho$ and $\sigma$ do not commute, and $\rho^0=\sigma^0$; then 
\begin{align*}
f(0)=\DU(\rho\|\sigma)<\DBS(\rho\|\sigma)=f(1),
\end{align*}
where the strict inequality is due to \cite[Example 4.4]{HM2016}. Hence, $f$ is not constant on $[0,1]$.
This proves \ref{item:gamma-relentr mon2}.

The statement in \ref{item:gamma-relentr mon3} follows as 
$\DU(\rho\|\sigma)\le f(0)$; hence, if $\DU(\rho\|\sigma)=f(\gamma)$ for some $\gamma\in(0,1]$, then 
$f$ is constant $\DU(\rho\|\sigma)$ on $(0,\gamma]$, and therefore on the whole of $[0,1]$, according to 
\ref{item:gamma-relentr mon1}. As a consequence, $\DU(\rho\|\sigma)=f(1)=\DBS(\rho\|\sigma)$, which cannot hold
if $\rho$ and $\sigma$ do not commute, according to \cite[Example 4.4]{HM2016}.

Finally, consider \ref{item:gamma-relentr mon4}. By definition,
\begin{align}
D^{\Um,\#_0}(\rho\|\sigma)-\DU(\rho\|\sigma)
&=
\Tr\rho\left[\logn\sigma-\logn(\acc{\sigma}{\rho})\right]
=
\Tr\rho\left[\logn((\acc{\sigma}{\rho})\inv)-\logn\sigma\inv\right]\nn\\
&=
\Tr\rho\left[\rho^0\bz\logn(\rho^0\sigma\inv\rho^0)\jz \rho^0-\rho^0(\logn\sigma\inv)\rho^0\right]
\ge 0\label{eq:gamma-relentr mon proof1}
\end{align} 
where the third equality follows from \eqref{eq:ac explicit2}, and the inequality is due to 
Lemma \ref{lemma:logn opJensen}, according to which 
\begin{align}\label{eq:gamma-relentr mon proof2}
\rho^0\bz\logn(\rho^0\sigma\inv\rho^0)\jz \rho^0-\rho^0(\logn\sigma\inv)\rho^0\ge 0.
\end{align} 
Since $\rho$ is invertible on $\rho^0$, equality in 
\eqref{eq:gamma-relentr mon proof1} holds if and only if 
\eqref{eq:gamma-relentr mon proof2} holds with equality, 
which, again by Lemma \ref{lemma:logn opJensen}, is equivalent to $\rho^0\sigma\inv=\sigma\inv\rho^0$,
which in turn is equivalent to $\rho^0\sigma=\sigma\rho^0$.
\end{proof}

\subsection{Barycentric R\'enyi divergences}
\label{sec:barycentric}

In \cite{mosonyi2022geometric}, a general method was given to define multivariate quantum R\'enyi divergences from quantum relative entropies, which we briefly recall here. 
For a finite set $\X$, 
fix quantum relative entropies $\D^{\qv}:=(D^{q_x})_{x\in\X}$, where we use the notation
$\qv:=(q_x)_{x\in\X}$.
For any probability distribution $P\in\S(\X)$, and any collection 
of PSD operators $W=(\ch{x})_{x\in\X}\in\B(\X,\hil)\p$, define
\begin{align}\label{eq:multirenyi def}
\DP{W}:=
\inf_{\hell\in\S(\hil)}\sum_{x\in\X}P(x)\D^{q_x}(\hell\|\ch{x}).
\end{align}
We call $\DP{W}$ the \ki{$P$-weighted barycentric R\'enyi-divergence of $W$ corresponding to $D^{\qv}$}.
When $q_x=q$ for some fixed $q$ and every $x\in\X$, we will use the notation 
$D_P^{\bary,q}=D_P^{\bary,\qv}$.

\begin{rem}
Since we only consider quantum relative entropies with the property that 
$D^q(\rho\|\sigma)<+\infty$ $\iff$ $\rho^0\le\sigma^0$, the optimization in 
\eqref{eq:multirenyi def} can be restricted to states $\omega$ such that $\omega^0\le \jsp{P}{W}$, 
whenever $\jsp{P}{W}\ne 0$,
where
\begin{align*}
\jsp{P}{W}:=\medwedge_{x\in\supp P}\ch{x}^0
\end{align*}
is the joint support of the $\ch{x}$, $x\in\supp P$. That is,
\begin{align*}
\DP{W}=
\begin{cases}
\inf_{\hell\in\S(\jsp{P}{W}\hil)}\sum_{x\in\X}P(x)\D^{q_x}(\hell\|\ch{x}),&\jsp{P}{W}\ne 0,\\
+\infty,&\jsp{P}{W}=0.
\end{cases}
\end{align*}
\end{rem}

According to \cite[Lemma 5.17]{mosonyi2022geometric}, $D^{\bary,\qv}$ is invariant under isometries, i.e., for any $W\in\B(\X,\hil)\p$ and any isometry $V:\,\hil\to\kil$, 
$\DP{(V\ch{x}V^*)_{x\in\X}}=\DP{W}$.
By \cite[Lemma 5.18]{mosonyi2022geometric}, if all $D^{q_x}$ are monotone then 
$D^{\bary,\qv}$ is a quantum extension of the multivariate classical R\'enyi divergence
\eqref{eq:classical multirenyi def} in the sense that 
if all $W_x$ are diagonal in some orthonormal basis $(e_i)_{i\in\I}$ then 
\begin{align*}
\DP{W}=-\log\Tr\prod_{x\in\supp(P)}\ch{x}^{P(x)}
=
-\log\sum_{i\in\I}\prod_{x\in\supp(P)}\inner{e_i}{\ch{x}e_i}^{P(x)}.
\end{align*}

\begin{rem}
It is customary to introduce some normalization in the definition of the quantum R\'enyi divergences, e.g., as 
$\wtilde{D}_P^{\bary,\qv}(W):=(1-\max_x P(x))\inv\DP{W}$, or 
$\what{D}_P^{\bary,\qv}(W):=(1-\max_x P_x)\inv\DP{\bz\frac{W_x}{\Tr W_x}\jz_{x\in\X}}$, 
when $P(x)<1$ and $\Tr W_x> 0$ for every $x\in\supp P$.
These play a role when considering limits where one or more $P(x)$ goes to $0$, which we do not consider here.
Regarding the additivity properties considered in Section \ref{sec:additivity},
the choice of normalization is irrelevant for (weak) additivity, while
the considerations regarding (strong) cq-additivity only work 
with the definition \eqref{eq:multirenyi def} and not with the above alternative normalizations.
\end{rem}

Lemma \ref{lemma:sandwich} yields immediately the following:
\begin{lemma}\label{lemma:barycentric sandwich}
For any finite set $\X$, any $P\in\S(\X)$, and 
any weakly subadditive and monotone quantum relative entropies $D^{q_x}$, $x\in\supp P$, 
\begin{align}\label{eq:minmax bary}
D_P^{\bary,\Um}(W)\le\DP{W}\le D_P^{\bary,\max}(W)
\end{align}
for any $W\in\B(\X,\hil)\p$. 
\end{lemma}

In general, there is no known explicit formula for $\DP{W}$ or the optimal $\omega$ in 
\eqref{eq:multirenyi def}. The only exception that has been known so far is when $D^{q_x}=\DU$ for every $x\in\X$, 
which was discussed in \cite[Proposition 5.29]{mosonyi2022geometric}. In this case 
there exists a unique optimal $\omega$, given by 
\begin{align}\label{eq:bary Um optimal state}
\hell_P^{\bary,\Um}(W):=\frac{\GM_P^{\LE}(W)}{\Tr \GM_P^{\LE}(W)},\ds\ds\text{where}\ds\ds
\GM_P^{\LE}(W):=\jsp{P}{W}\exp\bz\sum_{x\in\X}P(x)\jsp{P}{W}(\logn \ch{x})\jsp{P}{W}\jz
\end{align}
is the $P$-weighted log-Euclidean geometric mean of $W$. 
Hence,
\begin{align}\label{eq:minimal explicit}
D_P^{\bary,\Um}(W)=-\log\Tr\GM_P^{\LE}(W).
\end{align}
The proof follows immediately from Lemma \ref{lemma:variational} below with $n=1$
by taking into account that $\DU(\omega\|\oll\hell)\ge 0$ with equality if and only if $\omega=\oll\hell$.
More precisely, the above are valid when $\jsp{P}{W}\ne 0$; otherwise any state $\omega$ is optimal and gives the value
$+\infty=D_P^{\bary,\Um}(W)$. Note that this argument is different from the one given for the 
same statement in 
\cite[Proposition 5.29]{mosonyi2022geometric}.

In Proposition \ref{prop:sharp0 explicit} below we present a new example where the barycentric R\'enyi divergence
admits an explicit formula.

\begin{lemma}\label{lemma:variational}
Let $P\in\S(\X)$ and $W\in\B(\X,\hil)\p$ be such that $\jsp{P}{W}\ne 0$, and let 
$\oll\hell:=\hell_P^{\bary,\Um}(W)=\GM_P^{\LE}(W)/\Tr \GM_P^{\LE}(W)$.
For any state $\omega_n\in\S(\hil^{\otimes n})$ with 
$\omega_n^0\le (\jsp{P}{W})^{\otimes n}$,
\begin{align}\label{eq:variational}
\sum_{x\in\X}P(x)\DU(\omega_n\|\ch{x}^{\otimes n})
=
-n\log\Tr\GM_P^{\LE}(W)+\DU(\omega_n\|\oll\hell^{\otimes n}).
\end{align}
\end{lemma}
\begin{proof}
Let $\omega_n^{(k)}$ denote the $k$-th marginal of $\omega_n$, i.e., 
$\omega_n^{(k)}:=\Tr_{n\setminus\{k\}}\omega_n$.
Then 
\begin{align}
\sum_{x\in\X}P(x)\DU(\omega_n\|\ch{x}^{\otimes n})
&=
\sum_{x\in\X}P(x)\left[ \Tr\omega_n\logn\omega_n-\Tr\omega_n\logn\oll\hell^{\otimes n}+\Tr\omega_n\logn\oll\hell^{\otimes n}-\Tr\omega_n\logn\ch{x}^{\otimes n}
\right]\nn\\
&=
\DU(\omega_n\|\oll\hell^{\otimes n})+
\sum_{x\in\X}P(x)\left[\Tr\omega_n\logn\oll\hell^{\otimes n}-\Tr\omega_n\logn\ch{x}^{\otimes n}
\right]\nn\\
&=
\DU(\omega_n\|\oll\hell^{\otimes n})+
\sum_{x\in\X}P(x)\sum_{k=1}^n\left[\Tr\omega_n^{(k)}\logn\oll\hell-\Tr\omega_n^{(k)}\logn\ch{x}\right],
\label{eq:variational proof1}
\end{align}
where we used \eqref{eq:logn mult} and that 
$(\oll\hell^0)^{\otimes n}\omega_n(\oll\hell^0)^{\otimes n}=\omega_n$
and
$(\ch{x}^0)^{\otimes n}\omega_n(\ch{x}^0)^{\otimes n}=\omega_n$, $x\in\supp(P)$.

Using that $\logn\oll\hell=(-\log\Tr\GM_P^{\LE}(W))\jsp{P}{W}+\sum_{y\in\X}P(y)\jsp{P}{W}(\logn \ch{y})\jsp{P}{W}$ and that 
$(\omega_n^{(k)})^0\le \jsp{P}{W}$,  we get that 
\begin{align}
&\sum_{x\in\X}P(x)\left[\Tr\omega_n^{(k)}\logn\oll\hell-\Tr\omega_n^{(k)}\logn\ch{x}\right]\nn\\
&\ds=
-\log\Tr\GM_P^{\LE}(W)+
\underbrace{\sum_{x\in\X}P(x)\left[\Tr\omega_n^{(k)}\sum_{y\in\X}P(y)\jsp{P}{W}(\logn \ch{y})\jsp{P}{W}-\Tr\omega_n^{(k)}\logn\ch{x}\right]}_{=0}\,.
\label{eq:variational proof2}
\end{align}
Combining \eqref{eq:variational proof1} and \eqref{eq:variational proof2} yields
\eqref{eq:variational}.
\end{proof}

\begin{prop}\label{prop:sharp0 explicit}
If $D^{q_x}=D^{\Um,\#_0}$ for every $x\in\X$ then for any $P\in\S(\X)$ and 
$W\in\B(\X,\hil)\p$ with $S:=\jsp{P}{W}\ne 0$, there exists a unique optimal state in 
\eqref{eq:multirenyi def}, given by 
\begin{align}\label{eq:sharp 0 optimal state}
\omega_{P}^{\bary,\Um,\#_0}(W):=\frac{\GM_P^{\LE}(\acc{W}{S})}{\Tr \GM_P^{\LE}(\acc{W}{S})},\ds\ds\text{where}\ds\ds
\GM_P^{\LE}\bz\acc{W}{S}\jz=S\exp\bz\sum_{x\in\X}P(x)\logn (\acc{W_x}{S})\jz,
\end{align}
and $\acc{W}{S}:=(\acc{\ch{x}}{S})_{x\in\X}$,
and therefore
\begin{align}\label{eq:shorted minimal explicit}
D_P^{\bary,\Um,\#_0}(W)=-\log\Tr\GM_P^{\LE}(\acc{W}{S})=D_P^{\bary,\Um}\bz\acc{W}{S}\jz.
\end{align}
When $\jsp{P}{W}=0$, we have $D_P^{\bary,\Um,\#_0}(W)=+\infty$, and every state $\omega$ is optimal for
 \eqref{eq:multirenyi def}.
\end{prop}
\begin{proof}
The case $S=0$ is obvious, and hence for the rest we assume that $S\ne 0$. 
Let $\oll\hell:=\omega_{P}^{\bary,\Um,\#_0}(W)$. 
Then we have
\begin{align*}
-\log\Tr\GM_P^{\LE}(\acc{W}{S})
&=
D_P^{\bary,\Um}\bz \acc{W}{S}\jz\\
&=
\inf_{\omega\in\S(S\hil)}\sum_{x\in\X}P(x)\underbrace{\DU\bz\omega\|\acc{\ch{x}}{S}\jz}_{\le D^{\Um,\#_0}(\omega\|\ch{x})}\\
&\le
\inf_{\omega\in\S(S\hil)}\sum_{x\in\X}P(x)D^{\Um,\#_0}(\omega\|\ch{x})\\
&\le
\sum_{x\in\X}P(x)D^{\Um,\#_0}(\oll\omega\|\ch{x})\\
&=
\sum_{x\in\X}P(x)D^{\Um}\bz\oll\omega\|\acc{\ch{x}}{S}\jz\\
&=
-\log\Tr\GM_P^{\LE}(\acc{W}{S})\,.
\end{align*}
Here,
the first equality is due to \eqref{eq:minimal explicit}, the second equality is by definition, 
the first inequality follows from Lemma \ref{lemma:sharp0 bound}, 
the second inequality is trivial, 
the third equality follows by definition and the fact that 
$\oll\hell^0=S$, and the last equality follows by a straightforward computation.
Hence, $\oll\hell$ is an optimizer for $D_P^{\bary,\Um,\#_0}(W)$, and the equality in 
\eqref{eq:shorted minimal explicit} holds.

To see uniqueness, note that for any $\omega\in\S(S\hil)$, 
\begin{align*}
\sum_{x\in\X}P(x)D^{\Um,\#_0}(\omega\|\ch{x})
\ge
\sum_{x\in\X}P(x)D^{\Um}(\omega\|\acc{\ch{x}}{S})
=
-\log\Tr\GM_P^{\LE}(\acc{W}{S})
+\DU(\omega\|\oll\omega),
\end{align*}
where the first inequality follows from Lemma \ref{lemma:sharp0 bound}, and the equality 
is due to Lemma \ref{lemma:variational}. Since $\DU(\omega\|\oll\omega)\ge 0$ with equality if and only if 
$\omega=\oll\omega$, we get that no other $\omega$ attains $D_P^{\bary,\Um,\#_0}(W)$.
\end{proof}

\begin{prop}\label{prop:Um and sharp0 bary order}
For any finite set $\X$, any $P\in\S(\X)$, and any
$W\in\B(\X,\hil)\p$,
\begin{align}\label{eq:Um and sharp0 bary order}
D_{P}^{\bary,\Um}(W)
\le
D_{P}^{\bary,\Um,\#_0}(W),
\end{align}
and equality holds if and only if $S:=\jsp{P}{W}$ commutes with all 
$\ch{x}$, $x\in\supp P$. 
\end{prop}
\begin{proof}
If $S=0$ then it commutes with all $\ch{x}$, and 
$D_{P}^{\bary,\Um}(W)=+\infty=D_{P}^{\bary,\Um,\#_0}(W)$. Hence, for the rest we assume that 
$S\ne 0$. Let $\oll\omega:=\omega_{P}^{\bary,\Um,\#_0}(W)$ be as in 
\eqref{eq:sharp 0 optimal state}. Then 
\begin{align}
D_{P}^{\bary,\Um}(W)
\le
\sum_{x\in\X}P(x)\DU(\oll\omega\|\ch{x})
\le
\sum_{x\in\X}P(x)D^{\Um,\#_0}(\oll\omega\|\ch{x})
=
D_{P}^{\bary,\Um,\#_0}(W),
\label{eq:Um and sharp0 bary order proof1}
\end{align}
where the first inequality is by definition, and the second inequality follows from 
Proposition \ref{prop:gamma-relentr mon}.

If $S$ commutes with all $\ch{x}$, $x\in\supp P$, then 
$\acc{\ch{x}}{S}=S\ch{x}S$, $x\in\supp P$, and 
$\omega_{P}^{\bary,\Um}(W)=\omega_{P}^{\bary,\Um,\#_0}(W)$. Hence, the first inequality in 
\eqref{eq:Um and sharp0 bary order proof1} holds as an equality, and so does the second inequality, due to 
Proposition \ref{prop:gamma-relentr mon}, since $\oll\omega^0=S$.

Conversely, assume that 
$D_{P}^{\bary,\Um}(W)=D_{P}^{\bary,\Um,\#_0}(W)$, so that both inequalities in 
\eqref{eq:Um and sharp0 bary order proof1} hold as equalities.
Equality in the second inequality then yields that 
$\oll\omega^0=S$ commutes with all $\ch{x}$, $x\in\supp P$, according to 
Proposition \ref{prop:gamma-relentr mon}. This completes the proof.
\end{proof}

No explicit expression is known for $D_P^{\bary,\Um,\#_{\gamma}}(W)$ 
for any $\gamma\in(0,1]$
in general. The following observation
will be useful in deriving an explicit expression for the regularized version of $D_P^{\bary,\Um,\#_{\gamma}}$
in Section \ref{sec:regularized}.

\begin{lemma}\label{lemma:sharp restriction equality}
For any finite set $\X$, and any $P\in\S(\X)$ and $W\in\B(\X,\hil)\p$ with 
$S:=\jsp{P}{W}\ne 0$,
\begin{align*}
D_P^{\bary,\Um,\#_\gamma}(W)=D_P^{\bary,\Um,\#_\gamma}\bz\acc{W}{S}\jz,\ds\ds\ds\gamma\in[0,1],
\end{align*}
where $\acc{W}{S}:=(\acc{\ch{x}}{S})_{x\in\X}$.
\end{lemma}
\begin{proof}
We have
\begin{align*}
D_P^{\bary,\Um,\#_\gamma}\bz\acc{W}{S}\jz
&=
\inf_{\omega\in\S(S\hil)}\sum_xP(x)D^{\Um,\#_\gamma}(\omega\|\acc{W_x}{S})\\
&=
\inf_{\omega\in\S(S\hil)}\sum_xP(x)D^{\Um,\#_\gamma}(\omega\|\ch{x})
=
D_P^{\bary,\Um,\#_\gamma}(W),
\end{align*}
where the first and the last equalities are by definition, and the second equality is due to 
Lemma \ref{lemma:sharp0 bound}.
\end{proof}

\subsection{Additivity}
\label{sec:additivity}

For any finite set $\X$, quantum relative entropies $D^{\qv}=(D^{q_x})_{x\in\X}$, and probability distribution $P\in\S(\X)$, $D_P^{\bary,\qv}$ is an $\X$-variable quantum divergence as defined in Section \ref{sec:qdiv},
and hence its key properties like monotonicity and (weak) (sub)additivity can be defined as special cases of those discussed in 
Section \ref{sec:qdiv}. 
In particular, a barycentric quantum R\'enyi divergence $D_P^{\bary,\qv}$ is
\begin{itemize}
\item
\ki{weakly subadditive} if for any $W\in\B(\X,\hil)\p$ and any $n\in\bN$,  
$\DP{W^{\otimes n}}\le n\DP{W}$;
\item
\ki{weakly additive} if for any $W\in\B(\X,\hil)\p$ and any $n\in\bN$,  
$\DP{W^{\otimes n}}=n\DP{W}$;
\item
\ki{subadditive} if for any $W^{(1)}\in\B(\X,\hil^{(1)})\p$ and $W^{(2)}\in\B(\X,\hil^{(2)})\p$,
$\DP{W^{(1)}\otimes W^{(2)}}\le\DP{W^{(1)}}+\DP{W^{(2)}}$;
\item
\ki{additive} if for any $W^{(1)}\in\B(\X,\hil^{(1)})\p$ and $W^{(2)}\in\B(\X,\hil^{(2)})\p$,
$\DP{W^{(1)}\otimes W^{(2)}}=\DP{W^{(1)}}+\DP{W^{(2)}}$.
\end{itemize}

Note that for any fixed quantum relative entropy $D^q$, 
$D_P^{\bary,q}$ is defined for every finite set $\X$ and probability distribution $P\in\S(\X)$,
and we may define
\begin{align*}
D^{\bary,q}:=(D^{\bary,q}_P)_{P\in\S([d]),\,d\in\bN}.
\end{align*}
This encompasses $D^{\bary,q}_P$ for any finite set $\X$ and $P\in\S(\X)$, since for any 
bijection $\beta:\,\X\to[|\X|]$ and any $W\in\B(\X,\hil)\p$,  
$D_P^{\bary,q}(W)=D_{P\circ\beta\inv}^{\bary,q}(W\circ\beta\inv)$.

In this case, we may consider further types of additivity properties that are natural when we think of the arguments $W$ as generalized (non-normalized) classical-quantum channels. 
Indeed, motivated by the parallel use of classical-quantum channels, one may consider
\begin{align*}
W^{(1)}\chtimes W^{(2)}:=\bz\ch{x_1}^{(1)}\otimes \ch{x_2}^{(2)}\jz_{x_1\in\X^{(1)},x_2\in\X^{(2)}}
\end{align*}
for any $W^{(1)}\in\B(\X^{(1)},\hil^{(1)})\p$ and 
$W^{(2)}\in\B(\X^{(2)},\hil^{(2)})\p$. 
We say that 
$D^{\bary,q}$ is 
\begin{itemize}
\item
\ki{cq-additive} if for any finite sets $\X^{(k)}$, any 
probability distributions $P^{(k)}\in\S(\X^{(k)})$, 
and any $W^{(k)}\in\B(\X^{(k)},\hil^{(k)})\p$, $k=1,2$, 
\begin{align*}
D_{P^{(1)}\otimes P^{(2)}}^{\bary,q}(W^{(1)}\chtimes W^{(2)})
=
D_{P^{(1)}}^{\bary,q}(W^{(1)})+D_{P^{(2)}}^{\bary,q}(W^{(2)}),
\end{align*}
where $(P^{(1)}\otimes P^{(2)})(x_1,x_2):=P^{(1)}(x_1)P^{(2)}(x_2)$, $(x_1,x_2)\in\X^{(1)}\times\X^{(2)}$;
\item
\ki{strongly cq-additive} if for any finite sets $\X^{(k)}$, any 
$W^{(k)}\in\B(\X^{(k)},\hil^{(k)})\p$, $k=1,2$, and any probability distribution 
$P\in\S(\X^{(1)}\times\X^{(2)})$,
\begin{align*}
D_{P}^{\bary,q}(W^{(1)}\chtimes W^{(2)})
=
D_{P^{(1)}}^{\bary,q}(W^{(1)})+D_{P^{(2)}}^{\bary,q}(W^{(2)}),
\end{align*}
where $P^{(1)}$ and $P^{(2)}$ are the first and the second marginals of $P$, respectively.
\end{itemize}
In this setting we also say that $D^{\bary,q}$ is (weakly) additive if 
$D_P^{\bary,q}$ is (weakly) additive for any finite set $\X$ and any $P\in\S(\X)$. 
Obviously, strong cq-additivity of $D^{\bary,q}$ implies both 
cq-additivity and additivity; the latter can be seen by choosing
$P(x_1,x_2):=\delta_{x_1,x_2}P(x_1)$, $(x_1,x_2)\in\X\times\X$
for any $P\in\S(\X)$.

\begin{rem}
Note that, unlike 
in the usual notion of additivity, the sets
$\X^{(1)}$ and $\X^{(2)}$ are not necessarily the same in the definition of (strong) cq-additivity, 
and the probability distributions are not fixed, either.
In this sense, the notion of (strong) cq-additivity is a relation among different quantum divergences.
\end{rem}

\begin{prop}\label{prop:strong cq add}
$D^{\bary,\Um}$ and 
$D^{\bary,\Um,\#_0}$ are strongly cq-additive.
\end{prop}
\begin{proof}
Let $W^{(k)}\in\B(\X^{(k)},\hil^{(k)})\p$, $k=1,2$, and $P\in\S(\X^{(1)}\times\X^{(2)})$.
Let us define $P^{(1)}_x:=(W^{(1)}_{x})^0\otimes I_{\hil^{(2)}}$, $x\in\X^{(1)}$, and 
$P^{(2)}_y:=I_{\hil^{(1)}}\otimes (W^{(2)}_{y})^0$, $y\in\X^{(2)}$.
Then 
\begin{align}
\jsp{P}{W^{(1)}\chtimes W^{(2)}}
&=
\medwedge_{(x_1,x_2)\in\supp(P)}(W^{(1)}_{x_1}\otimes W^{(2)}_{x_2})^0\nn\\
&=
\medwedge_{(x_1,x_2)\in\supp(P)}P_{x_1}^{(1)}P_{x_2}^{(2)}\nn\\
&=
\medwedge_{(x_1,x_2)\in\supp(P)} \bz P_{x_1}^{(1)}\medwedge P_{x_2}^{(2)}\jz\nn\\
&=
\bz\medwedge_{x_1\in\supp(P^{(1)})} P_{x_1}^{(1)}\jz\medwedge
\bz\medwedge_{x_2\in\supp(P^{(2)})} P_{x_2}^{(2)}\jz\label{eq:strong cq add proof1}\\
&=
\bz\medwedge_{x_1\in\supp(P^{(1)})} (W_{x_1}^{(1)})^0\otimes I_{\hil^{(2)}}\jz
\bz\medwedge_{x_2\in\supp(P^{(2)})} I_{\hil^{(1)}}\otimes (W_{x_2}^{(2)})^0\jz\nn\\
&=
\bz\jsp{P^{(1)}}{W^{(1)}}\otimes  I_{\hil^{(2)}}\jz
\bz I_{\hil^{(1)}}\otimes\jsp{P^{(2)}}{W^{(2)}}\jz\nn\\
&=
\jsp{P^{(1)}}{W^{(1)}}\otimes\jsp{P^{(2)}}{W^{(2)}}.\nn
\end{align}
In the above, \eqref{eq:strong cq add proof1}
follows from the fact that the wedge of projections is the projection onto the intersection of the ranges, and
the commutativity and associativity of intersection, while the rest of the steps are fairly obvious.
Thus,
\begin{align*}
D_P^{\bary,\Um}(W^{(1)}\chtimes W^{(2)})=+\infty
&\ds\iff\ds
\jsp{P}{W^{(1)}\chtimes W^{(2)}}=0\\
&\ds\iff\ds
\jsp{P^{(1)}}{W^{(1)}}=0\ds\ds\text{or}\ds\ds\jsp{P^{(2)}}{W^{(2)}}=0\\
&\ds\iff\ds
D^{\bary,\Um}_{P^{(1)}}(W^{(1)})=+\infty\ds\ds\text{or}\ds\ds
D^{\bary,\Um}_{P^{(2)}}(W^{(2)})=+\infty,
\end{align*}
proving the desired additivity for $D_P^{\bary,\Um}$ when 
$\jsp{P}{W^{(1)}\chtimes W^{(2)}}=0$.

Assume therefore that $\jsp{P}{W^{(1)}\chtimes W^{(2)}}\ne 0$.
Then by $\jsp{P}{W^{(1)}\chtimes W^{(2)}}=\jsp{P^{(1)}}{W^{(1)}}\otimes\jsp{P^{(2)}}{W^{(2)}}$
established above, and using \eqref{eq:logn mult}, it is straightforward to see that  
\begin{align*}
\GM_P^{\LE}\bz W^{(1)}\chtimes W^{(2)}\jz
=\GM_{P^{(1)}}^{\LE}\bz W^{(1)}\jz\otimes \GM_{P^{(2)}}^{\LE}\bz W^{(2)}\jz,
\end{align*}
whence
\begin{align*}
\hell_P^{\bary,\Um}\bz W^{(1)}\chtimes W^{(2)}\jz
=
\hell_{P^{(1)}}^{\bary,\Um}\bz W^{(1)}\jz
\otimes
\hell_{P^{(2)}}^{\bary,\Um}\bz W^{(2)}\jz,
\end{align*}
and 
\begin{align}\label{eq:minimal additivity}
D_P^{\bary,\Um}\bz W^{(1)}\chtimes W^{(2)}\jz=
D^{\bary,\Um}_{P^{(1)}}\bz W^{(1)}\jz
+
D^{\bary,\Um}_{P^{(2)}}\bz W^{(2)}\jz.
\end{align}

This proves strong cq-additivity for $D^{\bary,\Um}$, and 
from this, the strong cq-additivity for $D^{\bary,\Um,\#_0}$ follows immediately by 
\eqref{eq:shorted minimal explicit} using \eqref{eq:acc mult}.
\end{proof}

\subsection{Regularized barycentric R\'enyi divergences}
\label{sec:regularized}

Apart from the two examples discussed in Proposition \ref{prop:strong cq add}, no general additivity property is known for any barycentric R\'enyi divergence, not even weak additivity of
$D_P^{\bary,\qv}$ for a specific $D^{\qv}$ and non-Dirac $P$. 
One of the main obstacles to deciding whether a barycentric R\'enyi divergence is additive is the lack of an explicit expression for it, apart from the special cases with generating relative entropies
$\DU$ and $D^{\Um,\#_0}$. In this case, it is natural to look at regularized quantities, which is what we are going to do in this section. 

We start with the following simple observation about the 
weak additivity property for regularized barycentric R\'enyi divergences:

\begin{prop}\label{prop:barycentric weak subadd}
Let $\X$ be a finite set, $P\in\S(\X)$, and $D^{\qv}=(D^{q_x})_{x\in\X}$ be such that 
$D^{q_x}$ is weakly subadditive for every $x\in\supp P$. Then 
$D_P^{\bary,\qv}$ is weakly subadditive, and its regularized version 
\begin{align}\label{eq:barycentric weak subadd}
\lreg{D}_P^{\bary,\qv}(W)
=\liminf_{n\to+\infty}\frac{1}{n}D_P^{\bary,\qv}(W^{\otimes n})
=\inf_{n\in\bN}\frac{1}{n}D_P^{\bary,\qv}(W^{\otimes n}),\ds\ds\ds
W\in\B(\X,\hil)\p,
\end{align}
is weakly additive. 
Moreover, 
if $D^{q_x}$ is subadditive for every $x\in\supp P$ then the liminf in \eqref{eq:barycentric weak subadd} is 
actually a limit for every $W\in\B(\X,\hil)\p$.
\end{prop}
\begin{proof}
The weak subadditivity of $D_{P}^{\bary,\qv}$ is straightforward to verify by definition, 
and thus the equality in \eqref{eq:barycentric weak subadd}, and
the assertion about weak additivity follow immediately from 
Lemma \ref{lemma:regularized div liminf}. 
Assume that $D^{q_x}$ is subadditive for every $x\in\supp P$. If 
$D_P^{\bary,\qv}(W)<+\infty$ then the liminf in \eqref{eq:barycentric weak subadd} is 
a limit, again by Lemma \ref{lemma:regularized div liminf}. 
If $D_P^{\bary,\qv}(W)=+\infty$ then $\jsp{P}{W}=0$, whence
$\jsp{P}{W^{\otimes n}}=\jsp{P}{W}^{\otimes n}=0$, and 
$D_P^{\bary,\qv}(W^{\otimes n})=+\infty$ for every $n\in\bN$, and therefore the liminf in 
\eqref{eq:barycentric weak subadd} is trivially a limit.
\end{proof}

\begin{cor}\label{cor:regularized sharp gamma weak add}
For any $\gamma\in(0,1]$, any finite set $\X$, any $P\in\S(\X)$, and any 
$W\in\B(\X,\hil)\p$,
\begin{align}
\breg{D}_P^{\bary,\Um,\#_\gamma}(W)
=\lim_{n\to+\infty}\frac{1}{n}D_P^{\bary,\Um,\#_\gamma}(W^{\otimes n}),\ds\ds\ds
W\in\B(\X,\hil)\p,
\end{align}
exists, and $\breg{D}_P^{\bary,\Um,\#_\gamma}$ is weakly additive.
\end{cor}
\begin{proof}
Follows immediately from Proposition \ref{prop:barycentric weak subadd} due to the fact that all
$D^{\Um,\#_{\gamma}}$, $\gamma\in(0,1]$, are additive.
\end{proof}

In what follows, we prove that 
\begin{align}\label{eq:reg plan1}
\breg{D}^{\bary,\meas}=D^{\bary,\Um},\ds\ds\ds
\breg{D}^{\bary,\Um,\#_{\gamma}}=D^{\bary,\Um,\#_0},\ds\ds \gamma\in[0,1],
\end{align}
and for any monotone quantum relative entropies $(D^{q_x})_{x\in\X}$ and any
$P\in\S(\X)$ and $W\in\B(\X,\hil)\p$,
\begin{align}\label{eq:reg plan2}
D_P^{\bary,\Um}(W)\le\lreg{D}_P^{\bary,\qv}(W)
\le \ureg{D}_P^{\bary,\qv}(W)\le D_P^{\bary,\Um,\#_0}(W).
\end{align}
Obviously, \eqref{eq:reg plan1} is much stronger than Corollary \ref{cor:regularized sharp gamma weak add}, and it gives strong cq-additivity of $\breg{D}^{\bary,\Um,\#_{\gamma}}$, according to 
Proposition \ref{prop:strong cq add}.

We start with proving the first equality in \eqref{eq:reg plan1}.

\begin{prop}\label{prop:regularized meas barycentric}
Let $\X$ be a finite set, $P\in\S(\X)$, and $W\in\B(\X,\hil)\p$. Then 
\begin{align*}
D_P^{\bary,\Um}(W)-\frac{d\log(n+1)}{n}
\le
\frac{1}{n}D_P^{\bary,\meas}(W^{\otimes n})
\le
D_P^{\bary,\Um}(W),\ds\ds\ds n\in\bN,
\end{align*}
where $d:=\dim\hil$, and
\begin{align*}
\breg{D}_P^{\bary,\meas}(W)
=\lim_{n\to+\infty}\frac{1}{n}D_P^{\bary,\meas}(W^{\otimes n})
=
D_P^{\bary,\Um}(W).
\end{align*}
\end{prop}
\begin{proof}
For any $x\in\supp P$, any $n\in\bN$, and any $\omega_n\in\S(\hil^{\otimes n})$, we have 
\begin{align}\label{eq:pinching for Dmeas}
\frac{1}{n}\DU(\omega_n\|W_x^{\otimes n})-\frac{d\log(n+1)}{n}
\le
\frac{1}{n}D^{\meas}(\omega_n\|W_x^{\otimes n})
\le
\frac{1}{n}\DU(\omega_n\|W_x^{\otimes n}),\ds\ds\ds n\in\bN;
\end{align}
see, e.g., \cite{HT14,HP}.
From this and the additivity of $D_P^{\bary,\Um}$, all assertions follow immediately.
\end{proof}

Next, we prove $\breg{D}^{\bary,\Um,\#_{\gamma}}=D^{\bary,\Um,\#_0}$, $\gamma\in[0,1]$. 
We start with the following.

\begin{lemma}\label{lemma:sharp restriction equality2}
For any finite set $\X$, and any $P\in\S(\X)$ and $W\in\B(\X,\hil)\p$ with 
$S:=\jsp{P}{W}\ne 0$,
\begin{align*}
\breg{D}_P^{\bary,\Um,\#_\gamma}(W)=\breg{D}_P^{\bary,\Um,\#_\gamma}\bz\acc{W}{S}\jz,\ds\ds\ds\gamma\in[0,1].
\end{align*}
\end{lemma}
\begin{proof}
It is easy to see that 
$\jsp{P}{W^{\otimes n}}=\jsp{P}{W}^{\otimes n}$, and that for $S:=\jsp{P}{W}$, 
\begin{align*}
\acc{\ch{x}^{\otimes n}}{S^{\otimes n}}=
(\acc{\ch{x}}{S})^{\otimes n}, \ds\ds\ds n\in\bN;
\end{align*}
according to \eqref{eq:acc mult}. Hence, the assertion follows from 
Lemma \ref{lemma:sharp restriction equality}.
\end{proof}

\begin{definition}\label{def:uss}
We say that a sequence of states $\univ_n\in\S(\hil^{\otimes n})$, $n\in\bN$, is a sequence of 
\ki{universal symmetric states} on a finite-dimensional Hilbert space $\hil$ if there exists a polynomial
$g_d$, where $d=\dim\hil$, such that 
\begin{itemize}
\item
for every $n\in\bN$, $\univ_n$ is symmetric and it commutes with every symmetric state on $\hil^{\otimes n}$;
\item
for every symmetric state $\omega\in\S_{\symm}(\hil^{\otimes n})$, 
$\omega\le g_d(n)\univ_n$.
\end{itemize}
\end{definition}

It is well known that for every finite-dimensional Hilbert space $\hil$, there exists a sequence of universal symmetric states with $g_d(n):=(n+1)^{\frac{d(d+1)-2}{2}}$; see, e.g., \cite{universalcq} or
\cite[Appendix A]{MO-cqconv}.
Note that the definition implies that $\Omega_n$ commutes with every symmetric operator on $\hil^{\otimes n}$; in particular,
\begin{align}\label{eq:uss unitary inv}
U^{\otimes n}\Omega_n(U^{\otimes n})^*=\Omega_n
\end{align}
for any unitary $U$ on $\hil$.
Moreover, since the maximally mixed state on $\hil^{\otimes n}$ is symmetric, the domination property implies that $\Omega_n$ is invertible.

\begin{lemma}\label{lemma:tilted states}
Let $(\univ_n)_{n\in\bN}$ be a sequence of universal symmetric states on a finite-dimensional Hilbert space $\hil$, and let $\omega_n\in\S_{\symm}(\hil^{\otimes n})$, $n\in\bN$, be a sequence of symmetric states. 
For every $n\in\bN$ and $t\in(0,1]$, define
\begin{align*}
\omega_{n,t}:=\frac{G_{(t,1-t)}^{\LE}(\omega_n,\univ_n)}{\Tr G_{(t,1-t)}^{\LE}(\omega_n,\univ_n)},
\ds\ds\ds\text{where}\ds\ds\ds
G_{(t,1-t)}^{\LE}(\omega_n,\univ_n):=\omega_n^t\univ_n^{1-t}
\end{align*}
is the (log-Euclidean) $t$-weighted geometric mean of $\omega_n$ and $\univ_n$. Then 
\begin{align}\label{eq:tilted states1}
\DU(\omega_{n,t}\|\omega_n)\le\frac{1-t}{t}\log g_d(n),\ds\ds\ds t\in(0,1], \,n\in\bN.
\end{align}
In particular, 
\begin{align}\label{eq:tilted states2}
\lim_{n\to+\infty}\frac{1}{n}\DU(\omega_{n,t}\|\omega_n)=0,\ds\ds\ds t\in(0,1].
\end{align}
\end{lemma}
\begin{proof}
Note that $\omega_{n,t}^0=\omega_n^0$. 
Let $c_{n,t}:=\Tr\omega_n^t\univ_n^{1-t}$. We have
\begin{align*}
\logn\omega_{n,t}
&=
t\logn\omega_n+(1-t)\omega_n^0(\logn\univ_n)\omega_n^0-\omega_n^0\log c_{n,t},
\end{align*}
whence
\begin{align}\label{eq:tilted states proof2}
\Tr\omega_{n,t}\logn\omega_{n,t}
&=
t\Tr\omega_{n,t}\logn\omega_n+(1-t)\Tr\omega_{n,t}\logn\univ_n-\log c_{n,t}.
\end{align}
A simple rearrangement gives
\begin{align*}
-\log c_{n,t}=(1-t)\DU(\omega_{n,t}\|\univ_n)+t\DU(\omega_{n,t}\|\omega_n)\ge t\DU(\omega_{n,t}\|\omega_n),
\end{align*}
whence
\begin{align}\label{eq:tilted states proof1}
\DU(\omega_{n,t}\|\omega_n)\le-\frac{1}{t}\log c_{n,t}.
\end{align}
Since $\omega_n\le g_d(n)\univ_n$, we have $\omega_n^{1-t}\le g_d(n)^{1-t}\univ_n^{1-t}$. Thus,
\begin{align*}
1\ge c_{n,t}\ge\Tr\omega_n^t\bz\omega_n^{1-t}g_d(n)^{t-1}\jz=g_d(n)^{t-1}, 
\end{align*}
where the first inequality follows by the H\"older inequality.
Substituting this back into \eqref{eq:tilted states proof1} gives \eqref{eq:tilted states1}, and 
\eqref{eq:tilted states2} follows immediately.
\end{proof}

\begin{rem}\label{rem:explicit finiten1}
Let $\Omega_n$, $n\in\bN$, be universal symmetric states with $g_d(n)=(n+1)^{p(d)}$, where 
$p(d):=d(d+1)/2-1$. Choosing then $t_n:=c\sqrt{\frac{\log(n+1)}{n}}$ in Lemma \ref{lemma:tilted states} 
with some $c\in(0,1)$ yields
\begin{align*}
\frac{1}{n}\DU(\omega_{n,t_n}\|\omega_n)
\le
p(d)\frac{1-t_n}{t_n}\frac{\log (n+1)}{n}
\le
p(d)\frac{1}{t_n}\frac{\log (n+1)}{n}
\le
\frac{p(d)}{c}\sqrt{\frac{\log(n+1)}{n}},
\end{align*}
giving a quantitative bound on the convergence of 
$\frac{1}{n}\DU(\omega_{n,t_n}\|\omega_n)$ to $0$.
\end{rem}

\begin{lemma}\label{lemma:BS-UM bound}
Let $\omega\in\S(\hil)$, let $n\in\bN$, $t\in(0,1]$, and let $\Omega_n\in\S(\hil^{\otimes n})$ be 
a positive definite state that commutes with every symmetric state. 
Define
\begin{align}\label{eq:BS-UM bound state}
\omega_{n,t}:=\frac{G_{(t,1-t)}^{\LE}(\omega^{\otimes n},\univ_n)}{\Tr G_{(t,1-t)}^{\LE}(\omega^{\otimes n},\univ_n)},
\ds\ds\ds\text{where}\ds\ds\ds
G_{(t,1-t)}^{\LE}(\omega^{\otimes n},\univ_n):=(\omega^{\otimes n})^t\univ_n^{1-t}.
\end{align}
For any $A\in\B(\hil)\p$ with $\omega^0\le A^0$, 
\begin{align}
0\le\DBS(\omega_{n,t}\|A^{\otimes n})-D^{\Um,\#_0}(\omega_{n,t}\|A^{\otimes n})
&\le
n\norm{\logn(\omega^{t/2} (\acc{A}{\omega^0})\inv\omega^{t/2})+\logn (\acc{A}{\omega^0})-t\logn\omega}_{\infty}\label{eq:BS-UM bound}\\
&=
n\norm{\logn(\omega^{t/2} A\inv\omega^{t/2})-\logn (\omega^0 A\inv \omega^0)-t\logn\omega}_{\infty}.
\label{eq:BS-UM bound2}
\end{align}
\end{lemma}
\begin{proof}
The equality in \eqref{eq:BS-UM bound2} is immediate from \eqref{eq:ac explicit2}.

By \eqref{eq:acc mult} and Lemma \ref{lemma:sharp0 bound},
\begin{align*}
\DBS(\omega_{n,t}\|A^{\otimes n})-D^{\Um,\#_0}(\omega_{n,t}\|A^{\otimes n})
&=
\DBS(\omega_{n,t}\|(\acc{A}{\omega^0})^{\otimes n})-D^{\Um,\#_0}(\omega_{n,t}\|(\acc{A}{\omega^0})^{\otimes n}).
\end{align*}
Hence, we may assume without loss of generality that $A=\acc{A}{\omega^0}$, or equivalently, that 
$A^0=\omega^0$, and prove 
\eqref{eq:BS-UM bound} under this assumption.
The first inequality in \eqref{eq:BS-UM bound} 
follows from the monotonicity of $[0,1]\ni\gamma\mapsto
D^{\Um,\#_{\gamma}}(\omega_{n,t}\|A^{\otimes n})$; see
Proposition \ref{prop:gamma-relentr mon}.

Next, note that 
\begin{align*}
\omega_{n,t}^{1/2}(A^{\otimes n})\inv\omega_{n,t}^{1/2}
=
\frac{1}{c_{n,t}}\Omega_n^{1-t}(\omega^{t/2} A\inv\omega^{t/2})^{\otimes n},
\end{align*}
where $c_{n,t}:=\Tr (\omega^{\otimes n})^t\Omega_n^{1-t}$, and we used that 
$\Omega_n$ commutes with any symmetric state. The support projection of this operator is 
$S_n:=(\omega^0)^{\otimes n}=(\omega^{\otimes n})^0=\omega_{n,t}^0$, and 
\begin{align*}
\logn\bz\omega_{n,t}^{1/2}(A\inv)^{\otimes n}\omega_{n,t}^{1/2}\jz
=
-S_n\log c_{n,t}+(1-t)S_n(\log\Omega_n)S_n
+\sum_{k=1}^nS_n\iota_k\bz\logn\bz\omega^{t/2} A\inv\omega^{t/2}\jz\jz,
\end{align*}
where $\iota_k:\,X\mapsto X\otimes(\otimes_{j\in[n]\setminus\{k\}}I)$, $X\in\B(\hil)$, is the canonical 
embedding of $\B(\hil)$ into the $k$-th tensor component in $\otimes_{j\in[n]}\B(\hil)$.
Thus,
\begin{align}
\DBS(\omega_{n,t}\|A^{\otimes n})
&=
\Tr\omega_{n,t}\logn\bz\omega_{n,t}^{1/2}(A\inv)^{\otimes n}\omega_{n,t}^{1/2}\jz\nn\\
&=
-\log c_{n,t}+(1-t)\Tr\omega_{n,t}\logn\Omega_n
+\sum_{k=1}^n\Tr\omega_{n,t}\iota_k\bz\logn(\omega^{t/2} A\inv\omega^{t/2})\jz,
\label{eq:BS-UM bound proof1}
\end{align}
where we used \eqref{eq:logn mult} and that $(\omega^{t/2} A\inv\omega^{t/2})^0=\omega^0$.
On the other hand, 
\begin{align}
\D^{\Um,\#_0}(\omega_{n,t}\|A^{\otimes n})
&=
\Tr\omega_{n,t}\logn\omega_{n,t}-
\Tr\omega_{n,t}\logn A^{\otimes n}\nn\\
&=-\log c_{n,t}+(1-t)\Tr\omega_{n,t}\logn\Omega_n+t\Tr\omega_{n,t}\logn\omega^{\otimes n}
-\Tr\omega_{n,t}\logn A^{\otimes n}\nn\\
&=
-\log c_{n,t}+(1-t)\Tr\omega_{n,t}\logn\Omega_n+\sum_{k=1}^n\Tr\omega_{n,t}\bz\iota_k(t\logn\omega-\logn A)\jz,
\label{eq:BS-UM bound proof2}
\end{align}
where we used the identity 
\eqref{eq:tilted states proof2}, and that $\logn\omega^{\otimes n}=S_n\sum_{k=1}^n\iota_k(\logn\omega)$, 
$\logn A^{\otimes n}=(A^0)^{\otimes n}\sum_{k=1}^n\iota_k(\logn A)$.
Subtracting \eqref{eq:BS-UM bound proof2} from \eqref{eq:BS-UM bound proof1} yields
\begin{align}
\DBS(\omega_{n,t}\|A^{\otimes n})
-D^{\Um,\#_0}(\omega_{n,t}\|A^{\otimes n})
&=\sum_{k=1}^n\Tr\omega_{n,t}^{(k)}\left[\logn(\omega^{t/2} A\inv\omega^{t/2})+\logn A-t\logn\omega\right]\nn\\
&\le
n\norm{\logn(\omega^{t/2} A\inv\omega^{t/2})+\logn A-t\logn\omega}_{\infty},
\label{eq:BS-UM bound proof3}
\end{align}
where $\omega_{n,t}^{(k)}:=\Tr_{[n]\setminus\{k\}}\omega_{n,t}$ is the $k$-th marginal of $\omega_{n,t}$, 
and the inequality follows by a simple H\"older inequality.
This completes the proof of \eqref{eq:BS-UM bound}.
\end{proof}

\begin{lemma}\label{lemma:sharp regularized finiten}
Let $\X$ be a finite set, $P\in\S(\X)$, and $W\in\B(\X,\hil)\p$. If 
$S:=\jsp{P}{W}\ne 0$ then define
\begin{align*}
\oll\hell:=\hell_P^{\bary,\Um}\bz\acc{W}{S}\jz
:=
\frac{\GM_P^{\LE}(\acc{W}{S})}{\Tr \GM_P^{\LE}\bz\acc{W}{S}\jz},\ds\ds\ds\ds
\GM_P^{\LE}\bz\acc{W}{S}\jz:=S\exp\bz\sum_{x\in\X}P(x)S(\logn \acc{\ch{x}}{S})S\jz\,,
\end{align*}
and if $S=0$ then let $\oll\hell$ be an arbitrary state.
For any $t\in(0,1]$ and $n\in\bN$, define $\omega_{n,t}$ as in \eqref{eq:BS-UM bound state}, with 
$\omega:=\oll\omega$ and an arbitrary sequence of universal symmetric states $(\Omega_n)_{n\in\bN}$. 
Then we have, for any $\gamma\in[0,1]$, any $n\in\bN$, and any $t\in(0,1]$,
\begin{align}
&D_P^{\bary,\Um,\#_0}(W)
\le
\frac{1}{n}D_P^{\bary,\Um,\#_{\gamma}}(W^{\otimes n})\label{eq:sharp regularized finiten0}\\
&\ds\le
\frac{1}{n}\sum_{x\in\X}P(x)D^{\Um,\#_{\gamma}}(\omega_{n,t}\|\ch{x}^{\otimes n})
\le
\frac{1}{n}\sum_{x\in\X}P(x)\DBS(\omega_{n,t}\|\ch{x}^{\otimes n})
\label{eq:sharp regularized finiten1}\\
&\ds\le
D_P^{\bary,\Um,\#_0}(W)
+\frac{1-t}{t}\frac{\log g_d(n)}{n}
+\sum_{x\in\X}P(x)\norm{\logn(\oll\omega^{t/2} (\acc{\ch{x}}{S})\inv\oll\omega^{t/2})+
\logn (\acc{\ch{x}}{S})-t\logn\oll\omega}_{\infty}.
\label{eq:sharp regularized finiten2}
\end{align}
\end{lemma}
\begin{proof}
If $S=0$ then all the quantities in \eqref{eq:sharp regularized finiten0}--\eqref{eq:sharp regularized finiten2} are equal to $+\infty$, and hence the inequalities hold trivially. Hence, for the rest we assume that
$S\ne 0$. 

The inequality in \eqref{eq:sharp regularized finiten0} follows from
Propositions \ref{prop:gamma-relentr mon} and \ref{prop:strong cq add}; according to the latter,
$D_P^{\bary,\Um,\#_0}(W^{\otimes n})=nD_P^{\bary,\Um,\#_0}(W)$ for every $n\in\bN$. 
The first inequality in \eqref{eq:sharp regularized finiten1} is by definition, and the second inequality 
follows by Proposition \ref{prop:gamma-relentr mon}.
Hence, we only have to prove the inequality in \eqref{eq:sharp regularized finiten2}.
By Lemma \ref{lemma:BS-UM bound},
\begin{align*}
&\sum_{x\in\X}P(x)\DBS(\omega_{n,t}\|\ch{x}^{\otimes n})\\
&\ds\le
\sum_{x\in\X}P(x)D^{\Um,\#_0}(\omega_{n,t}\|\ch{x}^{\otimes n})
+n\sum_{x\in\X}P(x)\norm{\logn(\oll\omega^{t/2} (\acc{\ch{x}}{S})\inv\oll\omega^{t/2})+
\logn (\acc{\ch{x}}{S})-t\logn\oll\omega}_{\infty},
\end{align*}
where we used that $\omega_{n,t}^0=S^{\otimes n}$.
The first sum above can be upper bounded as
\begin{align*}
\sum_{x\in\X}P(x)D^{\Um,\#_0}(\omega_{n,t}\|\ch{x}^{\otimes n})
&=
\sum_{x\in\X}P(x)\DU(\omega_{n,t}\|(\acc{\ch{x}}{S})^{\otimes n})\\
&=
n\underbrace{D_P^{\bary,\Um}(\acc{W}{S})}_{=D^{\bary,\Um,\#_0}_P(W)}+\underbrace{\DU(\omega_{n,t}\|\oll\omega^{\otimes n})}_{\le\frac{1-t}{t}\log g_d(n)}
\le
n\left[D_P^{\bary,\Um,\#_0}(W)+\frac{1-t}{t}\frac{\log g_d(n)}{n}\right].
\end{align*}
where 
we used that $\omega_{n,t}^0=(\oll\omega^0)^{\otimes n}=S^{\otimes n}$ for the equalities, and 
Lemma \ref{lemma:tilted states} for the inequalities.
This completes the proof.
\end{proof}

\begin{cor}\label{cor:explicit finiten2}
In the setting of Lemma \ref{lemma:sharp regularized finiten}, choose
$\Omega_n$, $n\in\bN$, to be universal symmetric states with $g_d(n)=(n+1)^{p(d)}$, where 
$p(d):=d(d+1)/2-1$, and let $t_n:=c\sqrt{\frac{\log(n+1)}{n}}$, $n\in\bN$, 
with some $c\in(0,1)$. Then 
\begin{align}
&D_P^{\bary,\Um,\#_0}(W)
\le
\frac{1}{n}D_P^{\bary,\Um,\#_{\gamma}}(W^{\otimes n})
\le
\frac{1}{n}\sum_{x\in\X}P(x)D^{\Um,\#_{\gamma}}(\omega_{n,t_n}\|\ch{x}^{\otimes n})
\label{eq:sharp regularized finiten1-2}\\
&\ds\le
D_P^{\bary,\Um,\#_0}(W)
+\frac{p(d)}{c}\sqrt{\frac{\log(n+1)}{n}}
+\sum_{x\in\X}P(x)\norm{\logn(\oll\omega^{t_n/2} (\acc{\ch{x}}{S})\inv\oll\omega^{t_n/2})+
\logn(\acc{\ch{x}}{S})-t_n\logn\oll\omega}_{\infty},
\label{eq:sharp regularized finiten2-2}
\end{align}
\end{cor}
\begin{proof}
Immediate from Lemma \ref{lemma:sharp regularized finiten}.
\end{proof}

\begin{rem}\label{rem:explicit finiten bound}
In Corollary \ref{cor:explicit finiten2}, the last term in \eqref{eq:sharp regularized finiten2-2} can be further upper bounded by $\kappa\sqrt{\frac{\log(n+1)}{n}}$ with some $\kappa>0$ depending only on 
$W$, $P$, $\oll\omega$, and $c$. We will not need this later, but give a proof for completeness in Appendix 
\ref{sec:explicit finiten}. As a consequence, the upper bound in 
\eqref{eq:sharp regularized finiten2-2} can be further upper bounded to give
\begin{align}
\frac{1}{n}\sum_{x\in\X}P(x)D^{\Um,\#_{\gamma}}(\omega_{n,t_n}\|\ch{x}^{\otimes n})
\le
D_P^{\bary,\Um,\#_0}(W)
+\bz\frac{p(d)}{c}+\kappa\jz\sqrt{\frac{\log(n+1)}{n}},\ds\ds\ds n\in\bN.
\label{eq:sharp regularized finiten3-2}
\end{align}
\end{rem}

\begin{thm}\label{thm:sharp gamma regularized}
For every $\gamma\in[0,1]$, $\breg{D}^{\bary,\Um,\#_{\gamma}}=D^{\bary,\Um,\#_0}$. That is, for every 
finite set $\X$, every $P\in\S(\X)$, and every $W\in\B(\X,\hil)\p$, 
\begin{align}\label{eq:sharp gamma regularized}
\breg{D}_P^{\bary,\Um,\#_{\gamma}}(W)=D_P^{\bary,\Um,\#_0}(W)=-\log\Tr G_P^{\LE}\bz\acc{W}{\jsp{P}{W}}\jz.
\end{align}
Moreover, with $\omega_{n,t_n}$ as in Corollary \ref{cor:explicit finiten2},
\begin{align}\label{eq:sharp gamma regularized2}
\breg{D}_P^{\bary,\Um,\#_{\gamma}}(W)=
\lim_{n\to+\infty}\frac{1}{n}\sum_{x\in\X}P(x)D^{\Um,\#_{\gamma}}(\omega_{n,t_n}\|\ch{x}^{\otimes n}),
\end{align}
uniformly in $\gamma$.
\end{thm}
\begin{proof}
All the assertions are trivial when $S:=\jsp{P}{W}=0$, and hence for the rest we assume that $S\ne 0$.
The second equality in \eqref{eq:sharp gamma regularized} is due to Proposition \ref{prop:sharp0 explicit}.
By Corollary \ref{cor:explicit finiten2}, for every $\gamma\in[0,1]$ and every $n\in\bN$,
\begin{align}
0&\le \frac{1}{n}\sum_{x\in\X}P(x)D^{\Um,\#_{\gamma}}(\omega_{n,t_n}\|\ch{x}^{\otimes n})- D_P^{\bary,\Um,\#_0}(W)
\label{eq:sharp gamma regularized proof1}\\
&\le
\frac{p(d)}{c}\sqrt{\frac{\log(n+1)}{n}}
+\sum_{x\in\X}P(x)\norm{\logn(\oll\omega^{t_n/2} (\acc{\ch{x}}{S})\inv\oll\omega^{t_n/2})+
\logn (\acc{\ch{x}}{S})-t_n\logn\oll\omega}_{\infty},
\label{eq:sharp gamma regularized proof2}
\end{align}
where the upper bound in the last line is independent of $\gamma$. 
As $n\to+\infty$, the first term in 
\eqref{eq:sharp gamma regularized proof2} trivially goes to $0$, and so does the second term as well,
since for any $x\in\supp P$, 
\begin{align*}
\lim_{t\searrow 0}
\norm{\logn(\oll\omega^{t/2} (\acc{\ch{x}}{S})\inv\oll\omega^{t/2})+
\logn (\acc{\ch{x}}{S})-t\logn\oll\omega}_{\infty}
=
\norm{\logn(\acc{\ch{x}}{S})\inv+\logn(\acc{\ch{x}}{S})}_{\infty}=0.
\end{align*}
From this, all the statements of the theorem follow immediately.
\end{proof}

Finally, we prove \eqref{eq:reg plan2}.

\begin{lemma}\label{lemma:regularized bounds}
In the setting of Lemma \ref{lemma:sharp regularized finiten},
let $D^{\qv}=(D^{q_x})_{x\in\X}$ be such that 
$D^{q_x}$ is a monotone relative entropy for every $x\in\supp P$. Then
\begin{align}
&D_P^{\bary,\Um}(W)-\frac{d\log(n+1)}{n}\label{eq:finiten bound-1}\\
&\ds\le
\frac{1}{n}D_P^{\bary,\meas}(W^{\otimes n})
\le
\frac{1}{n}D_P^{\bary,\qv}(W^{\otimes n})\label{eq:finiten bound0}\\
&\ds\le
\frac{1}{n}\sum_{x\in\X}P(x)D^{q_x}(\omega_{n,t}\|\ch{x}^{\otimes n})
\le
\frac{1}{n}\sum_{x\in\X}P(x)\DBS(\omega_{n,t}\|\ch{x}^{\otimes n})\label{eq:finiten bound1}\\
&\ds\le
D_P^{\bary,\Um,\#_0}(W)
+\frac{1-t}{t}\frac{\log g_d(n)}{n}
+\sum_{x\in\X}P(x)\norm{\logn(\oll\omega^{t/2} (\acc{\ch{x}}{S})\inv\oll\omega^{t/2})+
\logn (\acc{\ch{x}}{S})-t\logn\oll\omega}_{\infty}\label{eq:finiten bound2}
\end{align}
for every $n\in\bN$ and every $t\in(0,1]$.
\end{lemma}
\begin{proof}
The first inequality in \eqref{eq:finiten bound0} follows from \eqref{eq:pinching for Dmeas},
and the second one from Lemma \ref{lemma:sandwich}. 
The first inequality in \eqref{eq:finiten bound1} 
is by definition, and the second one follows from Lemma \ref{lemma:sandwich}.
The inequality in \eqref{eq:finiten bound2} is immediate from Lemma \ref{lemma:sharp regularized finiten}.
\end{proof}

\begin{cor}\label{cor:monotone regularized bary finite}
In the setting of Lemma \ref{lemma:regularized bounds}, choosing $\omega_{n,t_n}$ as in Corollary \ref{cor:explicit finiten2} yields that for every $n\in\bN$,
\begin{align}
&D_P^{\bary,\Um}(W)-\frac{d\log(n+1)}{n}
\le
\frac{1}{n}D_P^{\bary,\qv}(W^{\otimes n})
\le
\frac{1}{n}\sum_{x\in\X}P(x)D^{q_x}(\omega_{n,t_n}\|\ch{x}^{\otimes n})
\label{eq:monotone regularized bary finite}\\
&\ds\le
D_P^{\bary,\Um,\#_0}(W)
+\frac{p(d)}{c}\sqrt{\frac{\log(n+1)}{n}}
+\sum_{x\in\X}P(x)\norm{\logn(\oll\omega^{t_n/2} (\acc{\ch{x}}{S})\inv\oll\omega^{t_n/2})+
\logn (\acc{\ch{x}}{S})-t_n\logn\oll\omega}_{\infty}.
\label{eq:monotone regularized bary finite2}
\end{align}
\end{cor}
\begin{proof}
Immediate from Lemma \ref{lemma:regularized bounds} and Corollary \ref{cor:explicit finiten2}.
\end{proof}

\begin{rem}\label{rem:Umegaki lower bound}
In the setting of Lemma \ref{lemma:regularized bounds}, if $D^{q_x}$ is also weakly subadditive for every 
$x\in\supp P$ then the inequalities in \eqref{eq:finiten bound-1}--\eqref{eq:finiten bound0}
can be improved to \begin{align*}
D^{\bary,\Um}_P(W)\le\frac{1}{n} D_P^{\bary,\qv}(W^{\otimes n}),
\end{align*}
 according to 
Lemma \ref{lemma:barycentric sandwich} and Proposition \ref{prop:strong cq add}.
\end{rem}

\begin{rem}\label{rem:explicit finiten bound2}
Similarly to Remark \ref{rem:explicit finiten bound},
the last term in \eqref{eq:monotone regularized bary finite2} can be further upper bounded by 
$\kappa\sqrt{\frac{\log(n+1)}{n}}$ with some $\kappa>0$ depending only on 
$W$, $P$, $\oll\omega$, and $c$; see Appendix \ref{sec:explicit finiten}. As a consequence,
the upper bound in \eqref{eq:monotone regularized bary finite2} can be further upper bounded to give
\begin{align}\label{eq:monotone regularized bary finite3}
\frac{1}{n}\sum_{x\in\X}P(x)D^{q_x}(\omega_{n,t_n}\|\ch{x}^{\otimes n})
\le
D_P^{\bary,\Um,\#_0}(W)+\bz \frac{p(d)}{c}+\kappa\jz\sqrt{\frac{\log(n+1)}{n}}.
\end{align}
\end{rem}

\begin{thm}\label{thm:main}
Let $\X$ be a finite set, let $P\in\S(\X)$, and 
$D^{\qv}=(D^{q_x})_{x\in\X}$ be such that 
$D^{q_x}$ is a monotone quantum relative entropy for every $x\in\supp P$.
Then for any $W\in\B(\X,\hil)\p$,
\begin{align}
-\log\Tr G_P^{\LE}(W)
=D_P^{\bary,\Um}(W)
&\le
\lreg{D}_P^{\bary,\qv}(W)
\le
\ureg{D}_P^{\bary,\qv}(W)
\le
D_P^{\bary,\Um,\#_0}(W)
=
-\log\Tr G_P^{\LE}\bz\acc{W}{S}\jz,\label{eq:main2}
\end{align}
where $S:=\jsp{P}{W}$. If $S$ commutes with all $\ch{x}$, $x\in\supp P$, then 
\begin{align}
-\log\Tr G_P^{\LE}(W)
=
\breg{D}_P^{\bary,\qv}(W)
=
\lim_{n\to+\infty}\frac{1}{n}\sum_{x\in\X}P(x)D^{q_x}(\omega_{n,t_n}\|\ch{x}^{\otimes n})
\label{eq:main4}
\end{align}
for $\omega_{n,t_n}$, $n\in\bN$, as in Corollary \ref{cor:monotone regularized bary finite}.
\end{thm}
\begin{proof}
All the statements in the theorem hold obviously if $S=0$, and hence for the rest we assume that $S\ne 0$. 
All the inequalities in \eqref{eq:main2} follow immediately from Corollary \ref{cor:monotone regularized bary finite}
by taking $n\to+\infty$ in 
\eqref{eq:monotone regularized bary finite}--\eqref{eq:monotone regularized bary finite2}, by the same argument as in the proof of 
Theorem \ref{thm:sharp gamma regularized}.
If $S$ commutes with all $\ch{x}$, $x\in\supp P$, then the same argument yields \eqref{eq:main4}, due to 
Proposition \ref{prop:Um and sharp0 bary order}.
\end{proof}

The above immediately yields that $D^{\bary,\Um}$ and $D^{\bary,\Um,\#_0}$ are the smallest and the largest ones, respectively, among the 
weakly additive barycentric R\'enyi divergences generated by monotone quantum relative entropies. 
That is, we have the following:

\begin{thm}\label{thm:main2}
For every finite set $\X$, every 
$P\in\S(\X)$, every 
$D^{\qv}=(D^{q_x})_{x\in\X}$ such that $D^{q_x}$ is monotone for every 
$x\in\supp P$, and every $W\in\B(\X,\hil)\p$, 
\begin{align}\label{eq:main sandwich wa}
D_P^{\bary,\Um}(W)\le D_P^{\bary,\qv}(W)\le D_P^{\bary,\Um,\#_0}(W),
\end{align}
provided that $D_P^{\bary,\qv}$ is weakly additive. Moreover, both inequalities 
in \eqref{eq:main sandwich wa} hold as
equalities  when $\jsp{P}{W}$ commutes with all $\ch{x}$, $x\in\supp P$.
\end{thm}

\begin{rem}
Note that in the case $\jsp{P}{W}\ne 0$, 
$1-t_n$ in the explicit upper bounds in Corollaries \ref{cor:explicit finiten2} and \ref{cor:monotone regularized bary finite}
is the weight of the universal symmetric state $\Omega_n$ 
in the definition of $\omega_{n,t_n}$, 
while $t_n$ is the weight of $\oll\hell^{\otimes n}$, the optimal state for 
$D_P^{\bary,\Um,\#_0}(W^{\otimes n})$.
Taking $t_n\to 0$ in the proofs of Theorems \ref{thm:sharp gamma regularized} and \ref{thm:main} can therefore 
be seen as asymptotically removing $\oll\hell^{\otimes n}$ from the suboptimal state
$\omega_{n,t_n}$. One might therefore wonder whether these proofs would work by using solely 
$\Omega_n$ as an asymptotically optimal sequence of states. This, however, is easily seen not to be the case.

Indeed, the first obvious obstacle is that $\Omega_n$ is invertible, and hence if 
$\jsp{P}{W}\lneq I$ then for any $(D^{q_x})_{x\in\X}$,
\begin{align*}
\frac{1}{n}\sum_{x}P(x)D^{q_x}(\Omega_n\|\ch{x}^{\otimes n})=+\infty,\ds\ds\ds n\in\bN,
\end{align*} 
while $D_P^{\bary,\Um,\#_0}(W)<+\infty$. 
This could be remedied by considering a sequence of universal symmetric states on 
$(\ran S)^{\otimes n}$, $n\in\bN$, instead.
This, however, still does not make $(\Omega_n)_{n\in\bN}$ asymptotically optimal except for very special cases, as we show below.
\end{rem}

\begin{prop}\label{prop:uss asymptotically optimal}
Let $\X$ be a finite set, and $P\in\S(\X)$, $W\in\B(\X,\hil)\p$ be such that 
$\jsp{P}{W}=I$. Let
$D^{\qv}=(D^{q_x})_{x\in\X}$ be such that $D^{q_x}$ is monotone for every 
$x\in\supp P$. Then the following are equivalent:
\begin{enumerate}
\item\label{item:uss aymptotically optimal1}
There exists a sequence of universal symmetric states $(\Omega_n)_{n\in\bN}$,
\begin{align}\label{eq:uss aymptotically optimal}
\lim_{n\to+\infty}\frac{1}{n}\sum_{x\in\X}P(x)D^{q_x}(\Omega_n\|\ch{x}^{\otimes n})=D_P^{\bary,\Um}(W).
\end{align}
\item\label{item:uss aymptotically optimal3}
There exists some $c\in\bR$ such that $\sum_{x\in\supp P}P(x)\log\ch{x}=cI$. 
\end{enumerate}
\end{prop}
\begin{proof}
Let $\oll\hell:=\omega^{\bary,\Um}_P(W)$ be the optimal state for 
$D_P^{\bary,\Um}(W)$, as in \eqref{eq:bary Um optimal state}.
By definition, the condition in \ref{item:uss aymptotically optimal3} is equivalent to 
$\oll\hell=I/d=:\tau$. 

Note that if $(\Omega_n)_{n\in\bN}$ is a sequence of universal symmetric states then 
it commutes with every $\ch{x}^{\otimes n}$ by definition, and hence for every $n\in\bN$,
\begin{align}
\sum_{x\in\X}P(x)D^{q_x}(\Omega_n\|\ch{x}^{\otimes n})
=
\sum_{x\in\X}P(x)\DU(\Omega_n\|\ch{x}^{\otimes n})
=
nD_P^{\bary,\Um}(W)+\DU(\Omega_n\|\oll\hell^{\otimes n}),
\label{eq:uss asymptotically optimal proof1}
\end{align}
where the second equality is due to Lemma \ref{lemma:variational}.
It follows immediately from \eqref{eq:uss unitary inv} that all marginals $\Omega_n^{(k)}:=\Tr_{[n]\setminus\{k\}}\Omega_n$ of 
$\Omega_n$ commute with every unitary $U\in\B(\hil)$, whence $\Omega_n^{(k)}=I/d=\tau$ for every $k\in[n]$. Thus, 
\begin{align*}
\DU(\Omega_n\|\oll\hell^{\otimes n})
&=
\Tr\Omega_n\log\Omega_n-\Tr\Omega_n\log\oll\hell^{\otimes n}\\
&=
\Tr\Omega_n\log\Omega_n-
\underbrace{\Tr\Omega_n\log\tau^{\otimes n}}_{=n\Tr\tau\log\tau}
+\underbrace{\Tr\Omega_n\log\tau^{\otimes n}}_{=n\Tr\tau\log\tau}
-\underbrace{\Tr\Omega_n\log\oll\hell^{\otimes n}}_{=n\Tr\tau\log\oll\hell}\\
&=
-H(\Omega_n)+n\log d+n\DU(\tau\|\oll\hell),
\end{align*}
where $H(\Omega_n)$ is the von Neumann entropy of $\Omega_n$.
Combining this with \eqref{eq:uss asymptotically optimal proof1} yields 
\begin{align}
0&\le
\DU(\tau\|\oll\hell)+\log d-\limsup_{n\to+\infty}\frac{1}{n}H(\Omega_n)\nn\\
&\le
\liminf_{n\to+\infty}\frac{1}{n}\sum_{x\in\X}P(x)D^{q_x}(\Omega_n\|\ch{x}^{\otimes n})
-D_P^{\bary,\Um}(W)\nn\\
&\le
\limsup_{n\to+\infty}\frac{1}{n}\sum_{x\in\X}P(x)D^{q_x}(\Omega_n\|\ch{x}^{\otimes n})
-D_P^{\bary,\Um}(W)\nn\\
&=
\DU(\tau\|\oll\hell)+\log d-\liminf_{n\to+\infty}\frac{1}{n}H(\Omega_n).
\label{eq:uss asymptotically optimal proof2}
\end{align}
Since $\DU(\tau\|\oll\hell)\ge 0$ with equality if and only if $\oll\hell=\tau$, and 
$\limsup_{n\to+\infty}\frac{1}{n}H(\Omega_n)\le\log d$, the above immediately yields that if 
$(\Omega_n)_{n\in\bN}$ is such that \eqref{eq:uss aymptotically optimal} holds then 
$\oll\hell=\tau$, proving 
\ref{item:uss aymptotically optimal1}$\imp$\ref{item:uss aymptotically optimal3}.

Conversely, assume that \ref{item:uss aymptotically optimal3} holds, and let 
$(\Omega_n)_{n\in\bN}$ be any sequence of universal symmetric states.
For any $n\in\bN$, define 
\begin{align*}
\tilde\Omega_n:=\bz 1-\frac{1}{n}\jz\tau^{\otimes n}+\frac{1}{n}\Omega_n.
\end{align*}
Then $\tilde\Omega_n$ is symmetric, it commutes with every symmetric operator on $\hil^{\otimes n}$, and if 
$\omega_n\le p_d(n)\Omega_n$ for every symmetric state $\omega_n$ on $\hil^{\otimes n}$ then 
$\omega_n\le \tilde p_d(n)\tilde\Omega_n$, where $\tilde p_d(n):=np_d(n)$ is a polynomial in $n$. Hence, 
$(\tilde\Omega_n)_{n\in\bN}$ is a sequence of universal symmetric states. By the concavity of the von Neumann entropy, 
\begin{align*}
n\log d\ge H(\tilde\Omega_n)\ge \bz 1-\frac{1}{n}\jz \underbrace{H(\tau^{\otimes n})}_{=n\log d}+\frac{1}{n}H(\Omega_n)\ge
(n-1)\log d,
\end{align*}
where the bound in the first inequality holds for the entropy of any state on $\hil^{\otimes n}$.
Hence, 
\begin{align}\label{eq:uss asymptotically optimal proof3}
\lim_{n\to+\infty}\frac{1}{n}H(\tilde\Omega_n)=\log d.
\end{align}
Writing out \eqref{eq:uss asymptotically optimal proof2} with $\tilde\Omega_n$ in place of $\Omega_n$ and using 
\eqref{eq:uss asymptotically optimal proof3} and that $\DU(\tau\|\oll\hell)=0$ by assumption, we get that 
\eqref{eq:uss aymptotically optimal} holds with $\tilde\Omega_n$ in place of $\Omega_n$.
This proves 
\ref{item:uss aymptotically optimal3}$\imp$\ref{item:uss aymptotically optimal1}.
\end{proof}

\subsection{Non-additive barycentric R\'enyi divergences}
\label{sec:nonadd}

Theorem \ref{thm:main} implies that barycentric R\'enyi divergences generated by quantum relative entropies 
satisfying a few basic properties are typically not even weakly additive. We have the following general 
result. 

\begin{prop}\label{prop:nonadd}
Let $\X$ be a finite set, $P\in\S(\X)$ be such that $|\supp P|\ge 2$, and 
for every $x\in\supp P$, let $D^{q_x}$ be weakly subadditive, monotone and lower semi-continuous, and 
assume that there exist two different $x_1,x_2\in\supp P$ such that 
for any non-commuting $\rho,\sigma\in\B(\hil)\pne$ with $\rho^0\le\sigma^0$, 
$\DU(\rho\|\sigma)<D^{q_{x_k}}(\rho\|\sigma)$, $k=1,2$. 
Assume that $W\in\B(\X,\hil)\pp$ is such that 
the intersection of the commutants of $W_{x_1}$ and of $W_{x_2}$
is $\bC I_{\hil}$, 
and that $\sum_xP(x)\jsp{P}{W}(\logn W_x)\jsp{P}{W}$ is not a scalar multiple of the identity operator.
Then
\begin{align*}
\breg{D}_P^{\bary,\qv}(W)< D_P^{\bary,\qv}(W).
\end{align*}
In particular, $D_P^{\bary,\qv}$ is not weakly additive.
\end{prop}
\begin{proof}
Note that there always exists a $W$ as in the statement whenever $\dim\hil\ge 2$.
Indeed, choose $W_{x_1}:=\sum_{k=1}^dk\pr{e_k}$,
$W_{x_2}:=\sum_{k=1}^dk\pr{f_k}$, where 
$(e_i)_{i\in[d]}$ and $(f_i)_{i\in[d]}$ are orthonormal bases on $\hil$ such that 
$|\inner{e_i}{f_j}|\in(0,1)$ for every $i,j\in[d]$ (one way to ensure this is by choosing
$(e_i)_{i\in[d]}$ and $(f_i)_{i\in[d]}$ to be mutually unbiased), and 
define $W_x:=I$, $x\in\X\setminus\{x_1,x_2\}$. Then 
any operator that commutes with both $W_{x_1}$ and $W_{x_2}$ would need to be diagonal in both bases
$(e_i)_{i\in[d]}$ and $(f_i)_{i\in[d]}$, which is only satisfied by scalar multiples of the identity.
Moreover, 
\begin{align*}
\sum_xP(x)\jsp{P}{W}(\logn W_x)\jsp{P}{W}=P(x_1)\logn W_{x_1}+P(x_2)\logn W_{x_2},
\end{align*}
and if this was equal to $cI$ 
for some $c\in\bR$ then it would imply that 
$W_{x_1}$ and $W_{x_2}$ commute, which contradicts our assumptions.


Let us now take any $W$ as in the statement.
According to Theorem \ref{thm:main} and Remark \ref{rem:Umegaki lower bound},
\begin{align}\label{eq:nonadd proof0}
\breg{D}_P^{\bary,\qv}(W)=D_P^{\bary,\Um}(W)\le D_P^{\bary,\qv}(W).
\end{align}
By the lower semi-continuity of the $D^{q_x}$ and the compactness of $\S(\hil)$, there exists an 
$\omega\in\S(\hil)$ such that 
\begin{align}\label{eq:nonadd proof1}
D_P^{\bary,\qv}(W)
=
\sum_{x\in\X}P(x)D^{q_x}(\omega\|W_x)
\ge
\sum_{x\in\X}P(x)\DU(\omega\|W_x)
\ge
D_P^{\bary,\Um}(W),
\end{align}
where the first inequality is due to Lemma \ref{lemma:sandwich}, 
and the second inequality is by definition.

Note that $\omega_P^{\bary,\Um}(W)=G_P^{\LE}(W)/\Tr G_P^{\LE}(W)$ is the unique optimizer for the
definition of $D_P^{\bary,\Um}(W)$, and it is not equal to $I/d$, due to the assumptions on $W$. Hence, if 
$\omega=I/d$ then the last inequality in \eqref{eq:nonadd proof1} is strict, and therefore so is 
the inequality in \eqref{eq:nonadd proof0} as well.

If $\omega\ne I/d$ then 
by the assumptions on $W$, 
$\omega$ does not commute with at least one of $W_{x_1}$ and $W_{x_2}$, and hence, by the assumptions on 
$D^{\qv}$, 
\begin{align*}
P(x_1)D^{q_{x_1}}(\omega\|W_{x_1})+P(x_2)D^{q_{x_2}}(\omega\|W_{x_2})
>
P(x_1)\DU(\omega\|W_{x_1})+P(x_2)\DU(\omega\|W_{x_2}).
\end{align*}
Hence, the first inequality in \eqref{eq:nonadd proof1} is strict, and therefore
the inequality in \eqref{eq:nonadd proof0} is strict as well.
\end{proof}

\begin{cor}\label{cor:sharp nonadd1}
The statement of Proposition \ref{prop:nonadd} holds if 
$D^{q_{x_1}}=D^{\Um,\#_{\gamma_1}}$,
$D^{q_{x_2}}=D^{\Um,\#_{\gamma_2}}$ for some $\gamma_1,\gamma_2\in(0,1]$.
\end{cor}
\begin{proof}
Immediate from Propositions \ref{prop:nonadd} and \ref{prop:gamma-relentr mon}.
\end{proof}

As a special case of Corollary \ref{cor:sharp nonadd1}, we get the following:
\begin{cor}
For any $\gamma\in(0,1]$, $D^{\bary,\Um,\#_{\gamma}}$ is not weakly additive, i.e., 
for any finite set $\X$ and any $P\in\S(\X)$ with $|\supp P|\ge 2$, 
$D_P^{\bary,\Um,\#_{\gamma}}$ is not weakly additive. 
\end{cor}

\section*{Acknowledgments}

This research was partially supported by the National Research, Development and Innovation Office
of Hungary (NKFIH) via the research grants K 146380 and EXCELLENCE 151342, and by the
Ministry of Culture and Innovation and the National Research, Development and Innovation
Office within the Quantum Information National Laboratory of Hungary (Grant No. 2022-2.1.1-
NL-2022-00004). 
%

\section*{AI statement}

ChatGPT-6 Astra was instrumental in obtaining the results of the paper. 
The author provided the original problem and possible proof strategies and references, verified and reorganized the proof steps, significantly expanded the scope of the results, which also involved many rounds of interactions with ChatGPT-6 Astra, and the author wrote the paper.
The author takes full responsibility for the content of the paper.

\appendix

\section{The absolutely continuous part}
\label{sec:ac}

Following the convention set out in Section \ref{sec:prelim}, all inverses of PSD operators below are 
generalized inverses.
The following characterization of the positivity of a $2\times 2$ block operator is well known; see, e.g., 
\cite{Albert1969Pseudoinverses}.

\begin{lemma}\label{lemma:block positivity}
Let $\hil$ and $\kil$ be finite-dimensional Hilbert spaces and let $A\in\B(\hil),\,B\in\B(\kil)$ and $X\in\B(\hil,\kil)$. For $Y:=\begin{bmatrix} A & X^*\\ X & B \end{bmatrix}\in\B(\hil\oplus\kil)$, the following are equivalent:
\begin{enumerate}
\item\label{Y positive}
$Y\ge 0$;
\item\label{positivity condition1}
$B\ge 0$, \ds $\ran X\subseteq \supp B$, \ds  and  \ds $A-X^*B\inv X\ge 0$;
\item\label{positivity condition2}
$A\ge 0$, \ds  $\ran X^*\subseteq\supp A$, \ds  and \ds  $B-XA\inv X^*\ge 0$.
\end{enumerate}
\end{lemma}
\medskip

The inverse is well known to be operator monotone decreasing on positive definite operators, while the generalized inverse is easily seen not to have this property on PSD operators. Instead, the natural generalization is the following:

\begin{lemma}\label{lemma:inverse opmongen2}
For any $A,B\in\B(\hil)\p$, the following equivalence holds:
\begin{align}\label{eq:inverse opmongen2}
A\le B\ds\ds\iff\ds\ds A^0\le B^0\ds\text{and}\ds A\inv\ge A^0B\inv A^0.
\end{align}
\end{lemma}
\begin{proof}
According to Lemma \ref{lemma:block positivity}, both conditions in \eqref{eq:inverse opmongen2} 
are equivalent to $\begin{bmatrix}A\inv & A^0 \\ A^0 & B\end{bmatrix}\ge 0$.
\end{proof}

\begin{lemma}
For any $A,B\in\B(\hil)\p$,
\begin{align}\label{eq:ac support restricted}
\acc{A}{B}=
\max\{C\in\B(\hil)\p:\,C\le A,\,C^0=A^0\wedge B^0\}.
\end{align}
\end{lemma}
\begin{proof}
It is clear by definition that $C^0\le A^0\wedge B^0$ for any feasible solution $C$ in 
$\acc{A}{B}:=\max\{C\in\B(\hil)\p:\,C\le A,\,C^0\le B^0\}$, whence $(\acc{A}{B})^0\le A^0\wedge B^0$. 
Conversely, 
$A\ge \ep(A^0\wedge B^0)$ and $B^0\ge \ep(A^0\wedge B^0)$ for some $\ep>0$, whence 
$\ep(A^0\wedge B^0)$ is a feasible solution for the above optimization, and therefore 
$(\acc{A}{B})^0\ge (\ep(A^0\wedge B^0))^0=A^0\wedge B^0$.
From this, the claim follows immediately.
\end{proof}

\begin{prop}\label{prop:ac part special}
For any $A,B\in\B(\hil)\p$,
\begin{align}\label{eq:ac part special}
\acc{A}{B}=\bz (A^0\wedge B^0)A\inv(A^0\wedge B^0)\jz\inv.
\end{align}
\end{prop}
\begin{proof}
Let $P:=A^0\wedge B^0$, and note that 
\begin{align}\label{eq:ac part special proof1}
\bz PA\inv P\jz^{0}=P=\bz\bz PA\inv P\jz\inv\jz^{0}.
\end{align}
Consider any $C\in\B(\hil)\p$ that is feasible for the optimization in \eqref{eq:ac support restricted}. 
Then by Lemma \ref{lemma:inverse opmongen2},
$C\le A$ implies
\begin{align*}
C\inv\ge C^0A\inv C^0=PA\inv P,
\end{align*}
and hence again by Lemma \ref{lemma:inverse opmongen2},
\begin{align*}
C=(C\inv)\inv\le C^0\bz PA\inv P\jz\inv C^0=\bz PA\inv P\jz\inv,
\end{align*}
where we used that $C^0=P=\bz (PA\inv P)\inv\jz^0$. Thus,
\begin{align}\label{eq:ac part special proof4}
\acc{A}{B}\le\bz PA\inv P\jz\inv.
\end{align}

To see the converse inequality, note that 
\begin{align*}
I\ge \bz PA\inv P\jz^0
&=
\bz PA\inv P\jz^{-1/2}\bz PA\inv P\jz\bz PA\inv P\jz^{-1/2}\\
&=
\bz PA\inv P\jz^{-1/2}A\inv\bz PA\inv P\jz^{-1/2}=X^*X,
\end{align*}
where $X:=A^{-1/2}\bz PA\inv P\jz^{-1/2}$, and in the last equality we used 
that $P\bz PA\inv P\jz^{-1/2}=\bz PA\inv P\jz^{-1/2}$ due to 
\eqref{eq:ac part special proof1}.
Thus, $1\ge \norm{X}=\norm{X^*}$, which in turn yields
\begin{align*}
I\ge XX^*=
A^{-1/2}\bz PA\inv P\jz^{-1}A^{-1/2}, 
\end{align*}
and hence
\begin{align}\label{eq:ac part special proof2}
A\ge 
A^{0}\bz PA\inv P\jz^{-1}A^{0}
=
\bz PA\inv P\jz^{-1},
\end{align}
where we used again \eqref{eq:ac part special proof1} and $P\le A^0$ in the last equality.
Combining \eqref{eq:ac part special proof1} and \eqref{eq:ac part special proof2} yields that 
\begin{align}\label{eq:ac part special proof3}
\bz PA\inv P\jz^{-1}\le\acc{A}{B}.
\end{align}

Combining \eqref{eq:ac part special proof3} and \eqref{eq:ac part special proof4} yields
\eqref{eq:ac part special}.
\end{proof}

\begin{cor}\label{cor:ac mult}
For any $A_k,B_k\in\B(\hil^{(k)})\p$, $k=1,2$,
\begin{align*}
\acc{(A_1\otimes A_2)}{B_1\otimes B_2}=\acc{A_1}{B_1}\otimes\acc{A_2}{B_2}.
\end{align*}
\end{cor}
\begin{proof}
Immediate from Proposition \ref{prop:ac part special} as
$(A_1\otimes A_2)\inv=A_1\inv\otimes A_2\inv$, and
\begin{align*}
(A_1\otimes A_2)^0\wedge (B_1\otimes B_2)^0
&=
\bz(A_1^0\otimes I_2)\wedge (I_1\otimes A_2^0)\jz
\wedge
\bz(B_1^0\otimes I_2)\wedge (I_1\otimes B_2^0)\jz\\
&=
\bz(A_1^0\otimes I_2)\wedge (B_1^0\otimes I_2)\jz
\wedge
\bz(I_1\otimes A_2^0)\wedge (I_1\otimes B_2^0)\jz\\
&=
\bz (A_1^0\wedge B_1^0)\otimes I_2\jz\wedge
\bz I_1\otimes (A_2^0\wedge B_2^0)\jz\\
&=
\bz A_1^0\wedge B_1^0\jz\otimes \bz A_2^0\wedge B_2^0\jz.
\end{align*}
\end{proof}

\section{The proof of \eqref{eq:sharp regularized finiten3-2}}
\label{sec:explicit finiten}

For a non-zero PSD operator $A\in\B(\hil)\pne$, let $\lambda_{\min}^*(A)$ denote the smallest non-zero 
eigenvalue of $A$, and $\lambda_{\max}(A)=\norm{A}_{\infty}$ its largest eigenvalue.

The following may be well known; we give a proof for completeness.

\begin{lemma}\label{lemma:logdiff norm}
Let $A,B\in\B(\hil)\pne$ be such that $A^0=B^0=:P$. For any $\eta>0$ such that 
$A,B\ge \eta P$, 
\begin{align}
\norm{\logn A-\logn B}_{\infty}\le \eta\inv\norm{A-B}_{\infty}.
\end{align}
\end{lemma}
\begin{proof}
Recall that for a PSD operator $X\in\B(\hil)\p$, 
$\logn X=\int_0^{+\infty}\left[(1+s)\inv X^0-(X+sX^0)\inv\right]\,ds$. 
Thus,
\begin{align*}
\logn A-\logn B&=
\int_0^{+\infty}\left[(B+sP)\inv-(A+sP)\inv\right]\,ds
=
\int_0^{+\infty}(A+sP)\inv(A-B)(B+sP)\inv\,ds.
\end{align*}
Taking the norm then further yields
\begin{align*}
\norm{\logn A-\logn B}_{\infty}
&\le
\int_0^{+\infty}\norm{(A+sP)\inv}_{\infty}\norm{A-B}_{\infty}\norm{(B+sP)\inv}_{\infty}\,ds\\
&\le
\norm{A-B}_{\infty}\int_0^{+\infty}(\eta+s)^{-2}\,ds
=
\eta\inv\norm{A-B}_{\infty}.
\end{align*}
\end{proof}

\begin{lemma}\label{lemma:remainder bound}
In the setting of Lemma \ref{lemma:sharp regularized finiten}, 
\begin{align}
&\sum_{x\in\X}P(x)\norm{\logn(\oll\omega^{t/2} (\acc{\ch{x}}{S})\inv\oll\omega^{t/2})+
\logn (\acc{\ch{x}}{S})-t\logn\oll\omega}_{\infty}\nn\\
&\ds\le
sc\bz\log\frac{1}{\lambda_{\min}^*(\oll\omega)}\jz\times
\begin{cases}
1,& \jsp{P}{W}=0,\\
1+\lambda_{\min}^*(\oll\omega)^{-c}\sum_{x\in\supp P}P(x)
\frac{\lambda_{\max}(\ch{x})}{\lambda_{\min}^*(\ch{x})},&\jsp{P}{W}\ne 0.
\end{cases}
\label{eq:remainder bound}
\end{align}
for any $t,s,c\in(0,1)$ such that $t=sc$.
\end{lemma}
\begin{proof}
Let $S:=\jsp{P}{W}$. Then
\begin{align}
\nu&:=\sum_{x\in\X}P(x)\norm{\logn(\oll\omega^{t/2} (\acc{\ch{x}}{S})\inv\oll\omega^{t/2})+
\logn (\acc{\ch{x}}{S})-t\logn\oll\omega}_{\infty}\nn\\
&\le
-t\log\lambda_{\min}^*(\oll\omega)+
\sum_{x\in\X}P(x)
\norm{\logn(\oll\omega^{t/2} (\acc{\ch{x}}{S})\inv\oll\omega^{t/2})+
\logn (\acc{\ch{x}}{S})}_{\infty},
\label{eq:remainder bound proof0}
\end{align}
where the inequality follows by a simple triangle inequality.
If $S=0$ then the second term above is equal to $0$, and \eqref{eq:remainder bound}
follows trivially.

Assume for the rest that $S\ne 0$.
Then we may rewrite \eqref{eq:remainder bound proof0} as 
\begin{align}
\nu
&\le
-t\log\lambda_{\min}^*(\oll\omega)+
\sum_{x\in\X}P(x)\norm{\logn(\oll\omega^{t/2} \ch{x}\inv\oll\omega^{t/2})-
\logn (S\ch{x}\inv S)}_{\infty},
\label{eq:remainder bound proof1}
\end{align}
similarly to Lemma \ref{lemma:BS-UM bound}.

For every $x\in\supp P$, 
\begin{align*}
\oll\omega^{t/2} \ch{x}\inv\oll\omega^{t/2}
&\ge
\norm{\ch{x}}_{\infty}\inv\lambda_{\min}^*(\oll\omega)^t S,\\
S\ch{x}\inv S&\ge S(\norm{\ch{x}}_{\infty}\inv\ch{x}^0)S
=\norm{\ch{x}}_{\infty}\inv S\ge \norm{\ch{x}}_{\infty}\inv\lambda_{\min}^*(\oll\omega)^t S.
\end{align*}
Hence, by Lemma \ref{lemma:logdiff norm},
\begin{align*}
\norm{\logn(\oll\omega^{t/2} \ch{x}\inv\oll\omega^{t/2})-
\logn (S\ch{x}\inv S)}_{\infty}
\le
\norm{\oll\omega^{t/2} \ch{x}\inv\oll\omega^{t/2}-S\ch{x}\inv S}_{\infty}
\norm{\ch{x}}_{\infty}\lambda_{\min}^*(\oll\omega)^{-t},
\end{align*}
and
\begin{align*}
\norm{\oll\omega^{t/2} \ch{x}\inv\oll\omega^{t/2}-S\ch{x}\inv S}_{\infty}
&\le
\norm{\oll\omega^{t/2} \ch{x}\inv\oll\omega^{t/2}-\oll\omega^{t/2} \ch{x}\inv S}_{\infty}
+
\norm{\oll\omega^{t/2} \ch{x}\inv S-S\ch{x}\inv S}_{\infty}\\
&\le
\norm{\oll\omega^{t/2}-S}_{\infty}
\bz\norm{\oll\omega^{t/2} \ch{x}\inv}_{\infty}
+
\norm{\ch{x}\inv S}_{\infty}\jz\\
&\le 
-t\log\lambda_{\min}^*(\oll\omega)\norm{\ch{x}\inv}_{\infty}
=
\frac{-t\log\lambda_{\min}^*(\oll\omega)}{\lambda_{\min}^*(\ch{x})}\,.
\end{align*}
In the last inequality above, we used that $\norm{\oll\omega^{t/2}}_{\infty}\le 1$, and that
\begin{align*}
\norm{\oll\omega^{t/2}-S}_{\infty}
=\max_{\lambda\in\spec(\oll\omega)\setminus\{0\}}(1-\lambda^{t/2})
\le
\max_{\lambda\in\spec(\oll\omega)\setminus\{0\}}\bz-\frac{t}{2}\log\lambda\jz
=-\frac{t}{2}\log\lambda_{\min}^*(\omega).
\end{align*}
Writing this back into \eqref{eq:remainder bound proof1} yields
\begin{align}\label{eq:remainder bound proof2}
\nu\le-t\log\lambda_{\min}^*(\oll\omega)\left[1+
\lambda_{\min}^*(\oll\omega)^{-t}\sum_{x\in\supp P}P(x)\frac{\lambda_{\max}(\ch{x})}{\lambda_{\min}^*(\ch{x})}\right].
\end{align}
Finally, if $t=sc$ then $\lambda_{\min}^*(\oll\omega)^{-t}\le \lambda_{\min}^*(\oll\omega)^{-c}$
and \eqref{eq:remainder bound proof2} yields \eqref{eq:remainder bound}.
\end{proof}

\begin{cor}
The upper bounds in \eqref{eq:sharp regularized finiten3-2} and \eqref{eq:monotone regularized bary finite3}
hold with 
\begin{align*}
\kappa:=
c\bz\log\frac{1}{\lambda_{\min}^*(\oll\omega)}\jz\times
\begin{cases}
1,&\jsp{P}{W}=0,\\
1+\lambda_{\min}^*(\oll\omega)^{-c}\sum_{x\in\supp P}P(x)
\frac{\lambda_{\max}(\ch{x})}{\lambda_{\min}^*(\ch{x})},&\jsp{P}{W}\ne 0.
\end{cases}
\end{align*}
\end{cor}
\begin{proof}
Immediate from Lemma \ref{lemma:remainder bound} by choosing 
$s:=\sqrt{\frac{\log (n+1)}{n}}$.
\end{proof}

\bibliography{bibliography_br}

\end{document}